\documentclass[letter,12pt,fleqn]{amsart}
\usepackage[english]{babel}
\usepackage[utf8x]{inputenc}
\usepackage{amsmath}
\usepackage{amssymb}
\usepackage{amsthm}
\usepackage{mathtools}
\usepackage{fourier}
\usepackage{graphicx}
\usepackage{color}
\usepackage{subcaption}
\usepackage{algorithm}
\usepackage{algpseudocode}
\usepackage{algorithmicx}
\usepackage{enumitem}
\usepackage{url}

\usepackage[numbers]{natbib}
\usepackage[margin=1in]{geometry}
\usepackage{setspace}
\usepackage{letltxmacro}
\LetLtxMacro{\existsold}{\exists}
\renewcommand{\exists}{\existsold \hspace{.2em} }
\LetLtxMacro{\forallold}{\forall}
\renewcommand{\forall}{\: \forallold \: }
\LetLtxMacro{\intnolim}{\int}
\renewcommand{\int}{\intnolim\limits }
\newcommand{\dd}{\,\text{d}}
\renewcommand{\tilde}{\widetilde}

\newcommand{\norm}[1]{\left\lVert #1 \right\rVert }

\renewcommand{\bar}[1]{\overline{#1} }
\renewcommand{\Re}{\text{Re} }
\renewcommand{\Im}{\text{Im} }

\renewcommand{\epsilon}{\varepsilon }
\DeclareMathOperator*{\tr}{tr}
\renewcommand{\hat}{\widehat }

\newcommand{\bra}[1]{\left\langle#1\right|}
\newcommand{\ket}[1]{\left|#1\right\rangle}
\newcommand{\braket}[2]{\left\langle#1|#2\right\rangle}
\newcommand{\braopket}[3]{\left\langle#1\left|#2\right|#3\right\rangle}
\newcommand{\ketsmall}[1]{|#1\rangle}

\numberwithin{equation}{section} % numbers eqns by section
\newtheorem{lemma}{Lemma}[section]

\theoremstyle{definition}

\newcommand{\Pzeromean}{\mathbf{P}}
\newcommand{\sn}{\text{sn}}
\newcommand{\cn}{\text{cn}}
\newcommand{\dn}{\text{dn}}
\newcommand{\RPA}{\text{RPA}}
\newcommand{\occ}{\text{occ}}
\newcommand{\vir}{\text{vir}}
\newcommand{\aux}{\text{aux}}
\newcommand{\orb}{\text{orb}}
\newcommand{\ERPAcor}{E_\text{c}^{\RPA}}
\newcommand{\bJ}{\mathbf{J}}
\newcommand{\bK}{\mathbf{K}}
\newcommand{\orthP}{\mathbf{B}}

\usepackage[normalem]{ulem}

\usepackage{todonotes}

\newcommand{\changed}[1]{#1}
\usepackage{epstopdf}
\usepackage{epsfig}

\allowdisplaybreaks

\title[Cubic scaling RPA correlation via ISDF]{Cubic scaling algorithms for RPA correlation using interpolative separable density fitting}

\author{Jianfeng Lu}
\address{Department of Mathematics, Department of
  Physics, and Department of Chemistry, Duke University, Box 90320, Durham NC 27708, USA} \email{jianfeng@math.duke.edu}

\author{Kyle Thicke}
\address{Department of Mathematics, Duke University, Box 90320, Durham NC 27708, USA}
\email{kyle.thicke@duke.edu}

\date{\today} \thanks{This work is partially supported by
the National Science Foundation under grants DMS-1454939.\\
\copyright\, 2017.  This manuscript version is made available under the CC-BY-NC-ND 4.0 license\\ \texttt{http://creativecommons.org/licenses/by-nc-nd/4.0/}}

\begin{document}

\begin{abstract}
We present a new cubic scaling algorithm for the calculation of the RPA correlation energy.  Our scheme splits up the dependence between the occupied and virtual orbitals in $\chi^0$ by use of Cauchy's integral formula.  This introduces an additional integral to be carried out, for which we provide a geometrically convergent quadrature rule.  Our scheme also uses the newly developed Interpolative Separable Density Fitting algorithm to further reduce the computational cost in a way analogous to that of the Resolution of Identity method.
\end{abstract}

\maketitle

\section{Introduction}

In Density Functional Theory (DFT) \cite{HohenbergKohn:64, KohnSham:65}, the ground state energy of a many-body quantum system is written as a functional of the density $\rho$.  In the Kohn-Sham (KS) formalism of DFT \cite{KohnSham:65}, instead of considering the original interacting system of $N_\occ$ electrons, we consider a system of $N_\occ$ non-interacting electrons (the KS system) under a different external potential whose ground state density is identical to that of the interacting system.  In this effective single-particle system, the ground state density is given by
\begin{equation}
\rho(x) = \sum_{j=1}^{N_\occ} |\psi_j(x)|^2,
\end{equation}
where $\{\psi_j\}$ are the Kohn-Sham orbitals, the eigenstates of the effective single-particle system.  It is assumed throughout that the KS orbitals are ordered so that $\psi_1$ is the ground state of the KS system, $\psi_2$ is the first excited state, and so on.  In KS-DFT, the ground state energy of a system with $N_\occ$ interacting electrons can be written as
\begin{equation}
E = T_\text{s} + U_\text{ext} + U_\text{H} + E_\text{xc}, \label{GSenergy}
\end{equation}
where
\begin{align}
T_\text{s} &= \frac{1}{2}\sum_{j=1}^{N_\occ} \int |\nabla \psi_j(x)|^2 \dd x, & U_\text{ext} &= \int V_\text{ext}(x)\rho(x) \dd x, \\
U_\text{H} &= \frac{1}{2} \iint \rho(x)\rho(y)v(x,y) \dd x \dd y,
\end{align}
are, respectively, the kinetic energy of the effective single-particle
system, the potential energy due to the external potential
$V_\text{ext}$, and the so-called Hartree energy, which represents the
classical contribution of the energy from the Coulomb interaction
between electrons.  The remaining term in \eqref{GSenergy},
$E_\text{xc}$, is known as the exchange-correlation energy.  It has no
known simple form in terms of the density $\rho$ or the Kohn-Sham
orbitals $\{\psi_j\}$ and therefore needs to be approximated.  There
are many ways \cite{perdew2001jacob} of approximating this functional.
In this work, we consider one of the more accurate (and more computationally expensive) approximations, the
Random Phase Approximation (RPA).  In particular, we separate out the
exchange and correlation parts:
$E_\text{xc} = E_\text{x} + E_\text{c}$, and we use the exact exchange
$E_\text{x}^\text{EX}$ for the exchange energy $E_\text{x}$, and the
Random Phase Approximation to approximate the correlation energy
$E_\text{c}$ \cite{ren2012random}.
\begin{align}
E_\text{x}^\text{EX} &= -\sum_{jk} f_j f_k \iint \psi_j^*(x)\psi_k(x)\hat{v}(x,y)\psi_k^*(y)\psi_j(y) \dd x \dd y, \\
\ERPAcor &= \frac{1}{2\pi} \int_0^\infty \tr \left[ \ln(1-\hat{\chi}^0(i\omega)\hat{v}) + \hat{\chi}^0(i\omega)\hat{v} \right] \dd\omega, \label{RPAcorEnergy}
\end{align}
where
\begin{align}
\hat{\chi}^0(x,y,i\omega) &= \sum_{jk} \frac{(f_j-f_k)\psi_j^*(x)\psi_k(x)\psi_k^*(y)\psi_j(y)}{\epsilon_j - \epsilon_k - i\omega}, \label{chi_0} 
\end{align}
and $\hat{v}$ is the Coulomb kernel (in particular, we will consider the periodic Coulomb kernel, \changed{see Section~\ref{sec:RPAaux}}), and $\tr[A] = \int \braopket{x}{A}{x} \dd x$.  We will only consider the zero temperature case.  This means that, in the ground state, the first $N_\occ$ KS orbitals are filled while the rest are empty.  So, $f_\ell = 1$ if $1 \le \ell \le N_\occ$ (the occupied orbitals), and $f_\ell = 0$ if $\ell > N_\occ$ (the virtual orbitals).

In practice, one needs a way to obtain the KS orbitals before the energy can be computed.  This can be done via a self-consistent iteration.  However, we do not consider this here.  Instead, we only consider the calculation of the energy after the KS orbitals are known.  In this sense, we are considering a perturbative, non-self-consistent calculation of the RPA correlation energy.  That is, in a practical implementation, the KS orbitals could be calculated via a self-consistent iteration using a computationally less expensive, but also less accurate, approximation for the exchange-correlation energy functional (e.g., LDA, GGA).  The orbitals which are output from that self-consistent iteration can then be used to compute a more accurate approximation to the true correlation energy by using them to calculate the RPA correlation energy.  In this way, one obtains an approximation to the true correlation energy which is better than the less expensive method (LDA, GGA, etc.), but also does not require the self-consistent iteration to deal with the expense of RPA.

Of all the terms we have defined above, the RPA correlation energy
$\ERPAcor$ is the most computationally expensive to calculate.  The
goal of this paper is to provide a cubic scaling algorithm for the
computation of this term.  Before we can effectively talk about
scaling, we must first define some notation.  In this work, we will
use a spatial discretization with equally spaced grid points.  We
denote the total number of grid points by $n$.  We denote the number
of occupied orbitals, i.e., the number of electrons, by $N_\occ$.
Since there are infinitely many KS orbitals $\{\psi_j\}_{j=1}^\infty$
\footnote{When the spatial discretization is fixed with $n$ grid points, the
total number reduces to $n$, which is much larger than $N_{\occ}$.} and
the orbitals corresponding to higher energies will tend to have
smaller contributions to $\hat{\chi}^0$, we choose to use only the
$N_{\orb}$ KS orbitals of lowest energy in the RPA calculation.  This
gives us $N_\vir = N_\orb - N_\occ$ virtual orbitals.  Since $n$,
$N_\occ$, $N_\vir$, and $N_\orb$ all grow linearly with the system
size, we will sometimes refer to a general $N$ as a characterization
of the system size.

Very recently, a few cubic scaling methods for calculating the
RPA correlation energy have been presented in the literature.  The
general idea involved is to split up the $j$ and $k$ dependence in the
computation of $\hat{\chi}^0$ by introducing a new integral.  The idea
is easy to motivate.  From \eqref{RPAcorEnergy}, we can see that
everything can be done in cubic scaling if we are able to construct
the matrices $\hat{v}$ and $\hat{\chi}^0(i\omega)$ in cubic time.
$\hat{v}$ is not hard to construct, so we will focus on
$\hat{\chi}^0$.  $\hat{\chi}^0$ has $O(N^2)$ entries, so each entry
of $\hat{\chi}^0$ must be calculated in $O(N)$ time.  By inspection
of \eqref{chi_0}, it is clear that the most natural computation will
take $O(N^2)$ due to the coupling of $j$ and $k$ in the denominator.
But if we can decouple the $j$ and $k$ dependence, then we can sum
over each index separately and calculate each entry of $\hat{\chi}^0$
in $O(N)$.  Two different integrals have been utilized for this
purpose.  The first is
\begin{equation}
\changed{-\int_0^\infty e^{\epsilon_j t} e^{-\epsilon_k t} e^{-i\omega t} \dd t = \frac{1}{\epsilon_j - \epsilon_k - i\omega},}
\end{equation}
\changed{where $\epsilon_k > \epsilon_j$}.  Using this integral, one separates the dependence of $j$ and $k$ in \eqref{chi_0} into a product of exponentials inside the integral.  This leads to the Laplace transform cubic scaling methods.  This idea was first applied to RPA calculations in \cite{kaltak2014low} and \cite{kaltak2014cubic}, where a projector augmented wave (PAW) basis was utilized.  The idea was later extended to \changed{atomic orbitals \cite{ochsenfeld2016communication,  ochsenfeld2017vanishing, wilhelm2016large}}.  The other integral used to break up the $j$ and $k$ dependency is
\begin{equation}
\frac{1}{2\pi i}\int_\mathcal{C} \frac{1}{(\lambda - \epsilon_j + i\omega)(\lambda - \epsilon_k)} \dd\lambda = \frac{1}{\epsilon_j - \epsilon_k - i\omega},
\end{equation}
where $\mathcal{C}$ is a positively oriented closed contour that
encloses $\epsilon_j-i\omega$, but not $\epsilon_k$.  This idea was
first used in the context of cubic scaling RPA in
\cite{moussa2014cubic}.  Our algorithm in this paper will also adopt
this idea.

The most significant contribution of this paper is the reduction of the prefactor in front of the cubic term in the computational complexity.  In order to motivate this second main idea of this paper, let us examine how the density fitting (also called resolution of identity) approximation lowers the computational cost in the quartic scaling method \cite{eshuis2010fast, Ren_etal:12}.  The idea behind the approximation is that $\hat{\chi}^0$ is nearly equal to a low rank matrix due to its structure.  So, $\hat{\chi}^0$ (and $\hat{v}$) are formed into smaller sized matrices (by writing $\hat{v}$ into a smaller auxiliary basis and $\hat{\chi}^0$ into the dual basis) before the trace is taken.  The smaller matrix sizes lower the computational cost, but it is still $O(N^4)$ since the coupling of $j$ and $k$ is unaffected by the approximation.  

After we split up $j$ and $k$ in the denominator of \eqref{chi_0} using Cauchy's integral formula, we wish to further reduce the computational cost by taking advantage of the ``low rank'' nature of $\hat{\chi}^0$ using the same idea as density fitting (DF).  However, the DF approximation cannot be used in our case for two reasons.  The first is that the DF itself takes $O(N^4)$ operations, which destroys the cubic scaling.  The second problem is that $j$ and $k$ are coupled in the coefficients of the density fitting method.  As long as $j$ and $k$ remain coupled, $\hat{\chi}^0$ cannot be constructed in $O(N^3)$.  Solutions to both of these problems are provided by the interpolative separable density fitting (ISDF) method \cite{lu2015fast}.  The ISDF is capable of computing a decomposition of $\psi_j^*\psi_k$ similar to that of DF except that the $j$ and $k$ dependence in the coefficients are separated.  Additionally, the decomposition can be performed in $O(N^3)$ due to the use of a random projection in the method.  Our use of the ISDF also reduces the memory cost of our algorithm to $O(nN_{\aux})$ compared to $O(N_{\aux}^3)$ for the traditional resolution of identity approach.

Let us also mention the recent work \cite{lin2017adaptively}, where a
related problem of phonon calculation is approached from the point of
view of the Sternheimer equations to represent $\hat{\chi}$ acting on
functions.  Normally, for phonon calculations, $O(N^2)$ Sternheimer
equations would need to be solved in order to compute
$\hat{\chi}^0\hat{v}$.  However, by use of interpolative separable
density fitting and a polynomial interpolation, they reduce the number
of Sternheimer equations to $O(N)$, which enables a cubic scaling
algorithm.

The rest of the paper is organized as follows.  In Section \ref{sec:derivation}, we reformulate the expression \eqref{RPAcorEnergy} into a new, approximate form.  This new expression, characterized by \eqref{trtotr}, is used as the basis for our cubic scaling algorithm.  Section 3 begins with a summary of our algorithm followed by a detailed description of each step.  In Section \ref{sec:numerics}, we run numerical tests to compare the scaling of our algorithm against the quartic scaling resolution of identity algorithm.

\section{Derivation of the method}
\label{sec:derivation}

In this section, we reformulate the RPA correlation energy
\eqref{RPAcorEnergy} which will be used to construct our cubic scaling
algorithm. \changed{We will consider a computational domain with
  periodic boundary condition: Without loss of generality, up to a
  rescaling, the computational domain is taken to be $[0, 1]^d$, where
  $d$ is the spatial dimension. For simplicity, in this paper we just
  consider the $\Gamma$-point calculation with periodic Coulomb
  kernel, while we will leave Brillouin-zone sampling to future
  works. The interpolative separable density fitting has been extended
  to Bloch waves in \cite{lu2015fast}.}

\subsection{Contour integral representation} 

The first key idea is to split up the dependence of $j$
and $k$ in the denominator of \eqref{chi_0}.  This is accomplished
through the use of Cauchy's integral formula.  For a given $\omega$,
let $\mathcal{C}$ be a closed contour in the complex plane oriented in
the clockwise direction which encloses $\epsilon_\ell$ for all $\ell$
which are unoccupied, and does not enclose $\epsilon_\ell \pm i\omega$
for any $\ell$ that is occupied.  An example of such a contour is
shown in Figure \ref{fig:fullContour}.  
\begin{figure}[h]
\centering
\includegraphics[trim = 0in .5in 0in .8in, clip, scale=.8]{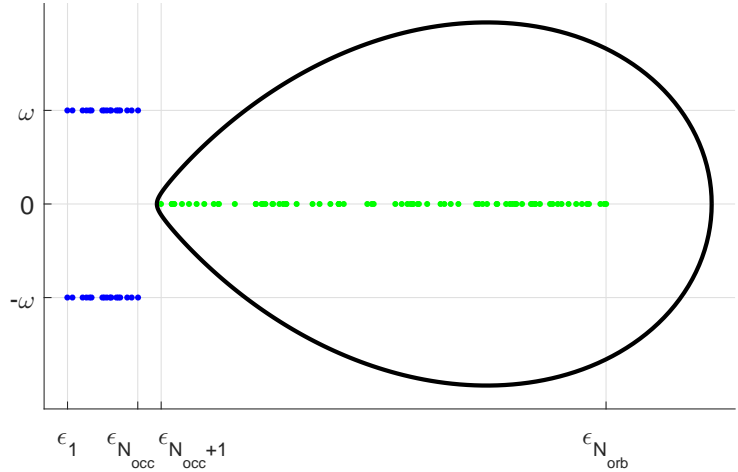}
\caption{An example of the contour $\mathcal{C}$ (see derivation in Section~\ref{subsec:contourIntegral}).  The blue points represent $\{\epsilon_j \pm i\omega\}_{j=1}^{N_\occ}$ (for a particular choice of $\omega$), and the green points represent $\{\epsilon_k\}_{k=N_{\occ}+1}^{N_\orb}$. } \label{fig:fullContour}
\end{figure}
While in principle the contour
can be chosen differently for different $\omega$, later in Section
\ref{subsec:contourIntegral}, we will make the restriction that
$\mathcal{C}$ is the same for all $\omega$ for the purpose of reducing
computational costs.  \changed{If $j$ is occupied and $k$ is unoccupied,} then using Cauchy's integral formula, we may write the coefficient in \eqref{chi_0} as
\begin{align}
\frac{f_j-f_k}{\epsilon_j - \epsilon_k - i\omega} &= \frac{1}{\epsilon_j - \epsilon_k - i\omega} \cdot \frac{1}{2\pi i}\int_\mathcal{C} \left(\frac{1}{\lambda - \epsilon_j + i\omega} - \frac{1}{\lambda - \epsilon_k}\right) \dd\lambda \notag\\
  &= \frac{1}{2\pi i}\int_\mathcal{C} \frac{1}{(\lambda - \epsilon_j + i\omega)(\lambda - \epsilon_k)} \dd\lambda. \label{CauchySplitjk}
\end{align}
This can then be used to obtain an expression for $\hat{\chi}^0$ which can be computed in cubic time,
\begin{align}
\braopket{x}{\hat{\chi}^0(i\omega)}{y} &= \sum_{jk} \frac{(f_j-f_k)\psi_j^*(x)\psi_k(x)\psi_k^*(y)\psi_j(y)}{\epsilon_j - \epsilon_k - i\omega}, \notag\\
  &= \sum_j^{\occ} \sum_k^{\vir} \frac{\psi_j^*(x)\psi_k(x)\psi_k^*(y)\psi_j(y)}{\epsilon_j - \epsilon_k - i\omega} + \text{c.c.} \notag\\
  &= \frac{1}{2\pi i} \int_\mathcal{C} \left(\sum_j^{\occ} \frac{\psi_j^*(x)\psi_j(y)}{\lambda - \epsilon_j + i\omega}\right) \left(\sum_k^{\vir} \frac{\psi_k(x)\psi_k^*(y)}{\lambda - \epsilon_k}\right) \dd\lambda + \text{c.c.}, \label{chi_0noISDF}
\end{align}
where $\text{c.c.}$ is the complex conjugate.  We show in Lemma
\ref{lem:contourConvRate} that the contour integral can be discretized
with a number of quadrature points which is logarithmic in
$(\epsilon_{N_\orb}-\epsilon_{N_\occ})/(\epsilon_{N_{\occ+1}}-\epsilon_{N_\occ})$.

Note that the formula \eqref{chi_0noISDF} already provides a cubic
scaling method for calculating $\chi^0(i\omega)$.  In particular,
ignoring logarithmic factors, $\chi^0(i\omega)$ can be calculated with
cost $O(N_{\orb}n^2)$.  However, the number of grid points $n$ could
be much larger than the number of orbitals, so to reduce the prefactor
of the computational cost, we will write the problem into an auxiliary
basis set instead of using the spatial grid points.  

Moreover, reducing the rank of $\hat{\chi}^0$ from
$N_\occ \cdot N_\vir$ (its rank before spatial discretization) to the
number of grid points $n$ (its rank after spatial discretization,
assuming $n < N_\occ\cdot N_\vir$) is really just a limitation placed
on the operator by the particular discretization of the problem,
rather than something inherent to the operator itself.
The intuitive idea
for the expected approximate low rank of $\hat{\chi}^0$ is that
$\hat{\chi}^0$ contains $O(N_\orb)$ information in its definition,
and therefore its approximate rank should scale linearly with the
number of \emph{orbitals} used in the calculation, rather than the
number of grid points.  This motivates the use of ISDF, as recalled in
the next subsection.

\subsection{Interpolative Separable Density Fitting}

We use the Interpolative Separable Density Fitting (ISDF)
\cite{lu2015compression,lu2015fast} to further accelerate the
computation. 
ISDF aims at a representation of the orbital pair functions as
\begin{equation}
  \psi_j^*(x)\psi_k(x) \approx \sum_{\mu=1}^{N_{\aux}} \psi_j^*(x_\mu) \psi_k(x_\mu) P_\mu(x), \label{ISDFexpansion}
\end{equation}
where the $\{x_\mu\}$ and $\{P_\mu\}$ are chosen by the ISDF
algorithm, which we will recall below for completeness of the
presentation.  Here $N_{\aux}$ denotes the number of auxiliary
orbitals needed to represent the orbital pairs involved; it is
empirically established that $N_{\aux}$ depends linearly on $N_{\orb}$
\cite{lu2015compression}, which we will also further verify in our
numerical examples.  The representation \eqref{ISDFexpansion} should
be compared to the traditional density fitting which yields
\begin{equation}
\psi_j^*(x)\psi_k(x) \approx \sum_{\mu=1}^{N_{\aux}} C_{jk}^\mu P_\mu(x),
\end{equation}
where $\{P_\mu(x)\}$ are inputs to the DF algorithm and the
coefficients $C_{jk}^\mu$ are determined via least squares fitting (in
the $L^2$ or Coulomb metric).  In \eqref{ISDFexpansion}, the
$\psi_j^*(x_\mu) \psi_k(x_\mu)$ factor is the coefficient for the
basis function $P_\mu(x)$. The main difference is thus that the $j$
and $k$ dependence of the coefficients are cleanly separated in ISDF,
but not in DF. This is important for achieving cubic scaling in
our work. Furthermore, ISDF has some other advantages over DF: the
time and memory cost of ISDF is cheaper, in particular, it requires
only $O(nN_{\aux})$ memory and cubic scaling computational cost. In
addition, the auxiliary basis functions are determined by the algorithm and
do not have to be specified by the user.

Let us now describe the ISDF algorithm.  The essential idea of the
ISDF algorithm is to select important grid points $\{x_{\mu}\}$ via a
randomized column selection algorithm. For the application to RPA
correlation energy, we only need orbital pair functions of the type
$\psi_j^*(x)\psi_k(x)$ where one of $j$ or $k$ is occupied and the
other is unoccupied (see \eqref{chi_0noISDF}).  We can use this to our
advantage by making a slight modification to the algorithm in
\cite{lu2015fast}, which would otherwise give an approximation for all
$N_\orb^2$ orbital pair functions.  A version of ISDF which only
calculates approximations for the orbital pair functions we are
interested in is presented in Algorithm~\ref{alg:ISDF}.
\begin{algorithm}[htp]
\caption{Interpolative Separable Density Fitting for RPA}
\label{alg:ISDF}
\begin{algorithmic}[1]
\Require{Orbitals $\{\psi_\ell\}_{\ell=1}^{N_\orb}$, error tolerance $\text{tol}$.}
\Ensure{$N_\aux$, $\{x_\mu\}$ and $\{P_\mu\}$ for $\mu = 1, ..., N_\aux$.}
\State Construct an $N_\occ \times n$ matrix $U^\occ$ such that the $j$th row of $U^\occ$ is $\psi_j$.  Likewise, construct an $N_\vir \times n$ matrix $U^\vir$ using the virtual orbitals as the rows.

\State Multiply each of $U^\occ$ and $U^\vir$ on the left by a random diagonal matrix, and then take the discrete Fourier transform,
\begin{align}
\hat{U}^\occ_\xi(x) &= \sum_{\alpha=1}^{N_\occ} e^{-i2\pi\alpha \xi / N_\occ} \eta_\alpha U_\alpha(x), \notag\\
\hat{U}^\vir_\xi(x) &= \sum_{\alpha=1}^{N_\vir} e^{-i2\pi\alpha \xi / N_\vir} \gamma_\alpha U_\alpha(x),
\end{align}
where $\eta_\alpha$ and $\gamma_\alpha$ are random unit complex numbers.

\State Randomly choose $r_\occ = c\sqrt{N_\occ}$ rows of $\hat{U}^\occ$ and $r_\vir = c\sqrt{N_\vir}$ rows of $\hat{U}^\vir$ to create submatrices $\mathcal{U}^\occ$ and $\mathcal{U}^\vir$, respectively.

\State Construct an $r_\occ r_\vir \times n$ matrix $M$,
\begin{equation}
M_{st}(x) = \bar{\mathcal{U}}_s^\occ(x) \mathcal{U}_t^\vir(x), \qquad s = 1,...,r_\occ, \quad t = 1,...,r_\vir,
\end{equation}
where $(st)$ is viewed as the row index of $M$.

\State Find the QR factorization with column pivoting (QRCP) of the $2r_\occ r_\vir \times n$ matrix $\mathbf{M}$ formed by concatenating $M$ with its complex conjugate,
\begin{equation}
QR = \left[\begin{array}{c} M \\ \text{conj}(M) \end{array}\right] E = \mathbf{M} E,
\end{equation}
where $Q$ is unitary, $R$ is upper triangular with diagonal entries in decreasing order, and $E$ is a permutation matrix.  In the case that $M$ is real, we can just take $\mathbf{M} = M$.

\State Choose $N_\aux$ such that 
\begin{equation}
R_{N_\aux,N_\aux} \ge \text{tol} \cdot R_{1,1} > R_{N_{\aux+1},N_{\aux+1}}.
\end{equation}
Then, we have $\mathbf{M} \approx (\mathbf{M}E)_{:,1:N_\aux} P$, where we note that $\mathbf{M}E$ is a permutation of the columns of $\mathbf{M}$.

\State Calculate $P = R_{1:N_\aux,1:N_\aux}^{-1} R_{1:N_\aux,:} E^T$, where MATLAB notation is used for the indexing.  Then, the auxiliary basis functions $\{P_\mu\}_{\mu=1}^{N_\aux}$ are the rows of $P$.

\State Finally, the points $\{x_\mu\}_{\mu=1}^{N_\aux}$ are the grid points used in the first $N_\aux$ columns of $\mathbf{M}E$.  We can be more specific as follows.  First, to avoid a conflict in notation, label the grid points $\{y_\ell\}_{\ell=1}^n$.  That is, whenever we considered $\psi_j$ as a row vector, it was $\psi_j = [\psi_j(y_1), ..., \psi_j(y_n)]$.  Next, since $E$ is a permutation matrix, it defines a permutation $\sigma$.  In particular, for a matrix $A$, the product $AE$ is a column permuted version of $A$ where the $j$th column of $A$ has been sent to the $\sigma(j)$ column.  Using this notation, we have $x_\mu = y_{\sigma^{-1}(\mu)}$ for $\mu = 1,..., N_\aux$.
\end{algorithmic}
\end{algorithm}

Compared with the original ISDF algorithm presented in \cite{lu2015fast}, one technical difference is Step 5 in Algorithm~\ref{alg:ISDF}. 
The reason for introducing $\mathbf{M}$, instead of just using $M$ in the QRCP there, is that now we can take the basis functions $\{P_\mu\}$ to be real.  This can be seen by switching $j$ and $k$ in \eqref{ISDFexpansion} and taking the complex conjugate,
\begin{equation}
\psi_j^*(x)\psi_k(x) \approx \sum_{\mu=1}^{N_{\aux}} \psi_j^*(x_\mu) \psi_k(x_\mu) P_\mu^*(x). \label{ISDFswitchedjk}
\end{equation}
Note that \emph{both} \eqref{ISDFexpansion} and \eqref{ISDFswitchedjk} are valid only because we included the conjugate of $M$ in the QRCP.  Now, we may average the two expressions to show that we may write $\psi_j^*(x)\psi_k(x)$ in terms of real basis functions,
\begin{equation}
\psi_j^*(x)\psi_k(x) \approx \sum_{\mu=1}^{N_{\aux}} \psi_j^*(x_\mu) \psi_k(x_\mu) \,\Re[P_\mu(x)]. \label{ISDFreal}
\end{equation}
It turns out that taking the basis functions to be real considerably simplifies the expression for $\chi^0$ that we will obtain.  This leads to reduced computational effort as well as simpler code.  For these reasons, we will always assume that the auxiliary basis functions from the ISDF are real.

\subsection{Representation of $\hat{\chi}^0$ using interpolative separable density fitting}

Now we can use the ISDF to reduce the computational cost of the cubic scaling method for the RPA correlation energy.  First, we approximate the $\hat{\chi}^0$ operator by an operator $\tilde{\chi}^0$ by using the ISDF approximation.  For simplicity of notation, we define $C_j^\mu = \psi_j(x_\mu)$. 
\begin{align}
\braopket{x}{\hat{\chi}^0(i\omega)}{y} &= \notag\\
  &\hspace{-5em}= \frac{1}{2\pi i} \int_\mathcal{C} \left(\sum_j^{\occ} \frac{\psi_j^*(x)\psi_j(y)}{\lambda - \epsilon_j + i\omega}\right) \left(\sum_k^{\vir} \frac{\psi_k(x)\psi_k^*(y)}{\lambda - \epsilon_k}\right) \dd\lambda \notag\\
  & + \frac{1}{2\pi i} \int_\mathcal{C} \left(\sum_j^{\occ} \frac{\psi_j(x)\psi_j^*(y)}{\lambda - \epsilon_j - i\omega}\right) \left(\sum_k^{\vir} \frac{\psi_k^*(x)\psi_k(y)}{\lambda - \epsilon_k}\right) \dd\lambda \notag\\
  &\hspace{-5em}\approx \sum_{\mu\nu} \frac{1}{2\pi i} \int_\mathcal{C} \left[ \left(\sum_j^{\occ} \frac{\bar{C}_j^\mu C_j^\nu}{\lambda - \epsilon_j + i\omega}\right) \left(\sum_k^{\vir} \frac{C_k^\mu \bar{C}_k^\nu}{\lambda - \epsilon_k}\right) \right. \notag\\
  & \qquad + \left. \left(\sum_j^{\occ} \frac{C_j^\mu \bar{C}_j^\nu}{\lambda - \epsilon_j - i\omega}\right) \left(\sum_k^{\vir} \frac{\bar{C}_k^\mu C_k^\nu}{\lambda - \epsilon_k}\right) \right] \dd\lambda \,\,P_\mu(x) P_\nu(y) \notag\\  
  &\hspace{-5em}= \sum_{\mu\nu} \frac{1}{2\pi i} \int_\mathcal{C} \Bigg[ \bJ_{\mu\nu}(\lambda,\omega) \bK_{\mu\nu}(\lambda) + \bar{\bJ_{\mu\nu}(\bar{\lambda},\omega)} \bar{\bK_{\mu\nu}(\bar{\lambda})} \Bigg] \dd\lambda \,\,P_\mu(x) P_\nu(y) \label{chitildeJK} \\
  &\hspace{-5em}= \sum_{\mu\nu} \chi^0_{\mu\nu}(i\omega) P_\mu(x) P_\nu(y) \notag\\
  &\hspace{-5em}= \braopket{x}{\tilde{\chi}^0(i\omega)}{y}, \label{chi_0inxy}
\end{align}
where the last line defines notation of $\tilde{\chi}^0$, and we have used the short hands
\begin{align}
\bJ_{\mu\nu}(\lambda,\omega) &= \sum_j^{\occ} \frac{\bar{C}_j^\mu C_j^\nu}{\lambda - \epsilon_j + i\omega}, \\
\bK_{\mu\nu}(\lambda) &= \sum_k^{\vir} \frac{C_k^\mu \bar{C}_k^\nu}{\lambda - \epsilon_k}, \\
\chi^0_{\mu\nu}(i\omega) &= \frac{1}{2\pi i} \int_\mathcal{C} \Bigg[ \bJ_{\mu\nu}(\lambda,\omega) \bK_{\mu\nu}(\lambda) + \bar{\bJ_{\mu\nu}(\bar{\lambda},\omega)} \bar{\bK_{\mu\nu}(\bar{\lambda})} \Bigg] \dd\lambda. \label{chi_0Def}
\end{align}
We note that in \eqref{chitildeJK}, the separability of the ISDF
coefficients into the $j$ and $k$ components is crucial.  Without this
separability (e.g., if a conventional DF was used) we would not be
able to calculate $\tilde{\chi}^0$ in cubic time since the sums over
$j$ and $k$ would not decouple.

Before continuing, we state explicitly our notation related to $\hat{\chi}^0$ for the sake of clarity:
\begin{itemize}
\item $\hat{\chi}^0$ -- the original operator.
\item $\tilde{\chi}^0$ -- the approximation to $\hat{\chi}^0$ that is obtained by applying the ISDF approximation.
\item $\chi^0$ -- defined by \eqref{chi_0Def}.  In \eqref{chi_0inDualBasis}, we
will show that it is $\tilde{\chi}^0$ in the dual basis to the auxiliary basis functions.
\item The argument $i\omega$ is often suppressed below for simplicity of notation.
\end{itemize}

\subsection{RPA correlation energy with auxiliary basis functions}\label{sec:RPAaux}

We now return our attention to \eqref{RPAcorEnergy}, \changed{where
  $\hat{v}$ is the periodic Coulomb kernel (see e.g., \cite{Lieb:81}),
 which is defined such that $g = \hat{v} f$ solves the Poisson equation
\begin{equation}
  - \Delta g = 4\pi  \Bigl(f - \int_{[0, 1]^d} f\Bigr)
\end{equation}
with periodic boundary conditions and such that $\int_{[0, 1]^d} g = 0$, to fix the arbitrary constant. $\hat{v}$ can be written as an integral operator,
\begin{equation}
  (\hat{v} f)(x) = \int_{[0, 1]^d} \hat{v}(x - y) g(y) \,\mathrm{d} y 
\end{equation}
where the kernel function, understood as a function in $L^2([0, 1]^d)$,
is given by
\begin{equation}\label{eq:periodicCoulomb}
  \hat{v}(x- y) = \frac{1}{\pi} \sum_{\xi \neq 0, \xi \in \mathbb{Z}^d} \frac{1}{\lvert\xi\rvert^2} e^{2\pi i \xi \cdot (x-y)}. 
\end{equation} 
} 
We want to rewrite the expression \eqref{RPAcorEnergy} using the
approximate basis $\left\{\ket{P_\mu}\right\}$.  However, since we are
considering a problem with periodic boundary conditions, it will be
advantageous for us to instead consider the basis
$\left\{\mathbb{1},
  \left\{\ket{\Pzeromean_\mu}\right\}_{\mu=1}^{N_{\aux}}\right\}$,
where $\ket{\Pzeromean_\mu}$ is the shift of $\ket{P_\mu}$ so that it
has zero mean, and $\mathbb{1}$ is the constant function with norm 1.
The reasons for this change are explained further in Section
\ref{subsec:step2}.  Since we wish to work with an orthonormal set, we
introduce the orthonormalized basis functions
\begin{equation}
\ketsmall{\orthP_\mu} = \ket{\Pzeromean_\nu} S_{\nu \mu}^{-1/2},
\end{equation}
where $S_{\mu\nu} = \braket{\Pzeromean_\mu}{\Pzeromean_\nu}$.  Then we can take the trace with respect to the orthonormal set 
\begin{equation*}
  \left\{\mathbb{1}, \left\{\ket{\orthP_\mu}\right\}_{\mu=1}^{N_{\aux}}\right\}.
\end{equation*}  
In the following, we consider $\hat{\chi}^0$ and $\hat{v}$ to be linear operators on an $n$-dimensional space (i.e., the discretization of the operators with respect to our spatial grid) for the purposes of justifying our steps.
\begin{align}
\tr\left[\ln(I-\hat{\chi}^0\hat{v})\right] &\approx \tr\left[\ln(I-\tilde{\chi}^0\hat{v})\right] \notag\\
  &\approx \sum_\beta \braopket{\orthP_\beta}{\ln(I-\tilde{\chi}^0\hat{v})}{\orthP_\beta} + \braopket{\mathbb{1}}{\ln(I-\tilde{\chi}^0\hat{v})}{\mathbb{1}}. \label{trDeriv1}
\end{align}
The first line is justified as follows.  First, we note that $\hat{v}$ is bounded on finite dimensional spaces.  Therefore, for $\tilde{\chi}^0$ close enough to $\hat{\chi}^0$, we have $\norm{\hat{\chi}^0\hat{v} - \tilde{\chi}^0\hat{v}} \ll 1$. 
  % Next, we assume that the eigenvalues of $\hat{\chi}^0\hat{v}$ are contained in $(-\infty, 0]$ \cite{cances2012mathematical, e2013kohn, cances2010dielectric}. \kt{not sure how relevant these refs are since we want for all $\omega$}
Thus, assuming that $I - \hat{\chi}^0 \hat{\nu}$ is invertible and $\ln ( I - \hat{\chi}^0 \hat{\nu})$ makes sense, the following linear approximation is justified
\begin{equation}
\tr\left[\ln(I-\hat{\chi}^0\hat{v}) - \ln(I - \tilde{\chi}^0\hat{v})\right] \approx -(I-\hat{\chi}^0\hat{v})^{-1} : (\hat{\chi}^0\hat{v} - \tilde{\chi}^0\hat{v}),
\end{equation}
where $A : B$ means $\sum_i\sum_j A_{ij} B_{ij}$, i.e., the sum of the entries of the entrywise product.  Before continuing our derivation, we first give a series expression for $\ln(I-\tilde{\chi}^0\hat{v})$.  To justify the expansion, we assume that the eigenvalues of $\hat{\chi}^0\hat{v}$ are contained in the left half complex plane.  Then for $\tilde{\chi}^0$ sufficiently close to $\hat{\chi}^0$, there exists $c > 1$ such that the following expansion holds.
\begin{align}
\ln(I-\tilde{\chi}^0\hat{v}) &=\ln\left[cI - \left((c-1)I+\tilde{\chi}^0\hat{v}\right)\right] \notag\\
  &= \ln(c)I + \ln\left[I - \frac{1}{c}\left((c-1)I+\tilde{\chi}^0\hat{v}\right)\right] \notag\\
  &= \ln(c)I - \sum_{\ell=1}^\infty \frac{[(c-1)I+\tilde{\chi}^0\hat{v}]^\ell}{\ell c^\ell} \notag\\
  &= \ln(c)I - \sum_{\ell=1}^\infty \frac{1}{\ell c^\ell} \sum_{p=0}^\ell {\ell \choose p} (c-1)^{\ell-p}(\tilde{\chi}^0\hat{v})^p. \label{logRewritten}
\end{align}

The first thing to note about \eqref{logRewritten} is that the
nullspace of $\hat{v}$ is contained in the nullspace of
$\ln(I-\tilde{\chi}^0\hat{v})$.  Since $\hat{v}$ is the periodic
Coulomb operator, the constant
function $\ket{\mathbb{1}}$ is in its nullspace \changed{(as $\xi = 0$ is excluded from the summation in \eqref{eq:periodicCoulomb})}.  Therefore, we may
drop the final term in \eqref{trDeriv1}, since it is zero. \changed{We
  remark that to accelerate the convergence with respect to the
  computational domain, a more sophisticated treatment of the Coulomb
  singularity at $\xi = 0$, rather than taking the periodic Coulomb
  kernel, is often used.  For example, see the method developed in
  \cite{Friedrich:2010}. }

Next, we note that the infinite sum in \eqref{logRewritten} is absolutely convergent, so we can continue \eqref{trDeriv1} by applying the above expansion and taking the trace through the sum.
\begin{align}
\tr\left[\ln(I-\hat{\chi}^0\hat{v})\right] &\approx \sum_\beta \braopket{\orthP_\beta}{\ln(I-\tilde{\chi}^0\hat{v})}{\orthP_\beta} \notag\\
  &= \sum_\beta \braopket{\orthP_\beta}{\ln(c)I}{\orthP_\beta} - \sum_{\ell=1}^\infty \frac{1}{\ell c^\ell} \sum_{k=0}^\ell {\ell \choose p} (c-1)^{\ell-p}\sum_\beta \braopket{\orthP_\beta}{(\tilde{\chi}^0\hat{v})^p}{\orthP_\beta}. \label{trDeriv2}
\end{align}
Next, we write the terms $\sum_\beta \braopket{\orthP_\beta}{(\tilde{\chi}^0\hat{v})^p}{\orthP_\beta}$ into a more computationally efficient, but approximate, representation.  For simplicity, we only give the derivation for $p = 1$, but the others are similar.
\begin{align}
\sum_\beta \braopket{\orthP_\beta}{\tilde{\chi}^0\hat{v}}{\orthP_\beta} &=\sum_\beta \left(\sum_\alpha S_{\beta\alpha}^{-1/2} \bra{\Pzeromean_\alpha}\right) \tilde{\chi}^0\hat{v} \left(\sum_\mu \ket{\Pzeromean_\mu} S_{\mu\beta}^{-1/2}\right) \notag\\
  &\approx \sum_{\alpha\beta\gamma\mu\nu} S_{\beta\alpha}^{-1/2} \braopket{\Pzeromean_\alpha}{\tilde{\chi}^0}{\Pzeromean_\gamma} S_{\gamma\nu}^{-1} \braopket{\Pzeromean_\nu}{\hat{v}}{\Pzeromean_\mu} S_{\mu\beta}^{-1/2} \notag\\
  &= \sum_{\alpha\gamma\mu\nu} S_{\mu\alpha}^{-1} \braopket{\Pzeromean_\alpha}{\tilde{\chi}^0}{\Pzeromean_\gamma} S_{\gamma\nu}^{-1} \braopket{\Pzeromean_\nu}{\hat{v}}{\Pzeromean_\mu} \notag\\
  &= \sum_{\mu\nu} \left(\sum_\alpha S_{\mu\alpha}^{-1} \bra{\Pzeromean_\alpha}\right)\tilde{\chi}^0\left(\sum_\gamma\ket{\Pzeromean_\gamma} S_{\gamma\nu}^{-1}\right) \braopket{\Pzeromean_\nu}{\hat{v}}{\Pzeromean_\mu} \notag\\
  &= \sum_{\mu\nu} \braopket{\tilde{\Pzeromean}_\mu}{\tilde{\chi}^0}{\tilde{\Pzeromean}_\nu} \braopket{\Pzeromean_\nu}{\hat{v}}{\Pzeromean_\mu}, \label{rewritingTrace}
\end{align}
where 
\begin{equation}
\bra{\tilde{\Pzeromean}_\mu} = \sum_\alpha S_{\mu\alpha}^{-1} \bra{\Pzeromean_\alpha},
\end{equation}
is the dual basis to $\left\{\ket{\Pzeromean_\mu}\right\}$.  Before commenting on the significance of this new representation, let's write $\braopket{\tilde{\Pzeromean}_\mu}{\tilde{\chi}^0}{\tilde{\Pzeromean}_\nu}$ into a simpler form.
\begin{align}
\braopket{\tilde{\Pzeromean}_\mu}{\tilde{\chi}^0}{\tilde{\Pzeromean}_\nu} &= \sum_{\alpha\gamma} S_{\mu\alpha}^{-1} \braopket{\Pzeromean_\alpha}{\tilde{\chi}^0}{\Pzeromean_\gamma} S_{\gamma\nu}^{-1} \notag\\
  &= \sum_{\alpha\gamma} S_{\mu\alpha}^{-1} \iint \braket{\Pzeromean_\alpha}{x}\braopket{x}{\tilde{\chi}^0}{y}\braket{y}{\Pzeromean_\gamma} \dd x \dd y \,\, S_{\gamma\nu}^{-1} \notag\\
  &= \sum_{\alpha\gamma\mu'\nu'} S_{\mu\alpha}^{-1} \iint \Pzeromean_\alpha(x) \Pzeromean_{\mu'}(x) \chi^0_{\mu'\nu'}(i\omega) \Pzeromean_{\nu'}(y) \Pzeromean_\gamma(y) \dd x \dd y \,\, S_{\gamma\nu}^{-1} \notag\\ 
  &= \sum_{\alpha\gamma\mu'\nu'} S_{\mu\alpha}^{-1} S_{\alpha\mu'} \chi^0_{\mu'\nu'}(i\omega) S_{\nu'\gamma} S_{\gamma\nu}^{-1} \notag\\ 
  &= \chi^0_{\mu\nu}(i\omega), \label{chi_0inDualBasis}
\end{align}
where $\chi^0_{\mu\nu}(i\omega)$ is defined in \eqref{chi_0Def}, and can therefore be computed in cubic time.  Let us define $v_{\mu\nu} = \braopket{\Pzeromean_\mu}{\hat{v}}{\Pzeromean_\nu}$.  Then the right hand side of \eqref{rewritingTrace} reads $\tr[\chi^0(i\omega)v]$, where the right hand side is just the standard trace of the product of the two matrices, $\tr[AB] = \sum_\mu\sum_\nu A_{\mu \nu}B_{\nu \mu}$.  Plugging this into \eqref{trDeriv2}, we obtain our final desired approximation,
\begin{equation}
\tr \left[ \ln(1-\hat{\chi}^0(i\omega)\hat{v}) + \hat{\chi}^0(i\omega)\hat{v} \right] \approx \tr \left[ \ln(1-\chi^0(i\omega)v) + \chi^0(i\omega)v \right]. \label{trtotr}
\end{equation}

\section{Algorithm}
\label{sec:algorithm}

In this section, we present the cubic scaling algorithm for the calculation of the RPA correlation energy.  We present a brief overview in Algorithm \ref{alg:cubicRPA} before going into the details of each step.  The computational effort is stated to the right of each step.

\begin{algorithm}[H]
\caption{Cubic scaling calculation of the RPA correlation energy}
\label{alg:cubicRPA}
\begin{algorithmic}[1]
\Require{Kohn-Sham orbitals $\{\psi_k\}$ and corresponding energies $\{\epsilon_k\}$.}
\Ensure{$\ERPAcor$}
\State Use $\{\psi_k\}_{k=1}^{N_{\orb}}$ as the input to ISDF to obtain $\{x_\mu\}_{k=1}^{N_{\aux}}$ and $\{P_\mu\}_{k=1}^{N_{\aux}}$. \hfill $O(nN_{\orb}^2)$
\State Compute the matrix $v_{\mu,\nu} = \braopket{\Pzeromean_\mu}{\hat{v}}{\Pzeromean_\nu}$. \hfill $O(nN_{\aux}^2)$
\State For each quadrature point $\omega_m$:
\begin{enumerate}[label=\alph*)]
\item Compute $\chi^0_{\mu,\nu}(i\omega_m) = \displaystyle\frac{1}{2\pi i} \int_\mathcal{C} \Bigg[ \bJ_{\mu,\nu}(\lambda,\omega_m) \bK_{\mu,\nu}(\lambda) + \bar{\bJ_{\mu,\nu}(\bar{\lambda},\omega_m)} \bar{\bK_{\mu,\nu}(\bar{\lambda})} \Bigg] \dd\lambda$. \hfill $O(N_{\orb}N_{\aux}^2)$
\item Compute $\frac{1}{2\pi} \tr\left[\ln(1-\chi^0(i\omega_m)v) + \chi^0(i\omega_m) v\right]$. \hfill $O(N_{\aux}^3)$
\end{enumerate}
\State Calculate $\ERPAcor = \frac{1}{2\pi} \int_0^\infty \tr\left[\ln(1-\chi^0(i\omega)v) + \chi^0(i\omega) v \right] \dd\omega$ via numerical quadrature.
\end{algorithmic}
\end{algorithm}

Without using the ISDF, the algorithm would be essentially exactly the same, but with Step 1 removed and Step 3a replaced by \eqref{chi_0noISDF}.  Except, in the computational costs, each $N_{\aux}$ would be replaced by $n$.  So clearly, if $N_{\aux}$ is much less than $n$, then including the ISDF can substantially speed up the algorithm.

\subsection{Step 2 -- Computing the Coulomb matrix}
\label{subsec:step2}

We note that $v$ can be efficiently computed by noticing that
\begin{align}
\braopket{\Pzeromean_\mu}{\hat{v}}{\Pzeromean_\nu} &= \iint \Pzeromean_\mu(x) \Pzeromean_\nu(y) v(x,y) \dd x \dd y \notag\\
	&= \int \Pzeromean_\mu(x) \phi_\nu(x) \dd x, \label{calculationOfv}
\end{align}
where 
\begin{equation}
\phi_\nu(x) = \int \Pzeromean_\nu(y) v(x,y) \dd y.
\end{equation}
Therefore, $\phi_\nu$ solves the Poisson equation, $-\Delta \phi_\nu = 4\pi \Pzeromean_\nu$ with periodic boundary conditions.  This equation can be efficiently solved using the fast Fourier transform.  Therefore, the functions $\phi_\nu$ can be precalculated at a total cost of $O(N_{\aux}n \log n)$.  Then the quadrature for \eqref{calculationOfv} is straightforward.

We have glossed over a couple details here.  First, the Poisson equation with periodic boundary conditions is not solvable unless $\Pzeromean_\nu$ has an average value of 0.  This is of course true by our construction of $\Pzeromean_\nu$, and this is the reason for the use of the $\{\ketsmall{\Pzeromean_\mu}\}$ basis rather than the $\{\ketsmall{P_\mu}\}$ basis when calculating the trace in \eqref{trDeriv1}.  The second detail we've skipped is that the solution $\phi_\nu$ is not unique as we can add any constant and get another solution.  However, it turns out that adding a constant to $\phi_\nu$ does not change the integral in \eqref{calculationOfv} since $\Pzeromean_\mu(x)$ has mean 0.  So, this problem is also avoided by the use of the $\{\ketsmall{\Pzeromean_\mu}\}$ basis rather than the $\{\ketsmall{P_\mu}\}$ basis.

\subsection{Step 3 -- Quadrature rule for the contour integral}
\label{subsec:contourIntegral}

Before discussing our proposed quadrature rule, let us discuss the symmetry of \eqref{chitildeJK}.  For notational purposes, define 
\begin{equation}
I_{\mu,\nu}(\lambda,\omega) = \frac{1}{2\pi i} \Bigg[ \bJ_{\mu,\nu}(\lambda,\omega) \bK_{\mu,\nu}(\lambda) + \bar{\bJ_{\mu,\nu}(\bar{\lambda},\omega)} \bar{\bK_{\mu,\nu}(\bar{\lambda})} \Bigg].
\end{equation}
It is straightforward to show the following symmetry across the real line, 
\begin{align}
\bar{\bK_{\mu,\nu}(\bar{\lambda})} &= \bK_{\nu,\mu}(\lambda), \\
\Re\left[I_{\mu,\nu}(\bar{\lambda},\omega)\right] &= - \Re\left[I_{\mu,\nu}(\lambda,\omega)\right], \label{symmAboutReal1} \\
\Im\left[I_{\mu,\nu}(\bar{\lambda},\omega)\right] &= \Im\left[I_{\mu,\nu}(\lambda,\omega)\right]. \label{symmAboutReal2}
\end{align}
We wish to calculate 
\begin{equation}
\chi^0_{\mu\nu}(i\omega) = \int_\mathcal{C} I_{\mu,\nu}(\lambda,\omega) \dd \lambda,
\end{equation}
where $\mathcal{C}$ is oriented clockwise and encloses $\{\epsilon_k\}_{k=N_{\occ}+1}^{N_{\orb}}$ while enclosing none of $\{\epsilon_j \pm i\omega\}_{j=1}^{N_{\occ}}$.  An example of such a contour is given in Figure \ref{fig:fullContour}.  In order to use the symmetry about the real axis, we will enforce our contour to be symmetric about the real axis.  Define $\mathcal{C}_\text{upper}$ and $\mathcal{C}_\text{lower}$ to be the parts of the contour in the upper and lower half plane.  Then using \eqref{symmAboutReal1} and \eqref{symmAboutReal2}, we have
\begin{align}
\chi^0_{\mu\nu}(i\omega) &= \int_{\mathcal{C}_\text{upper}} I_{\mu,\nu}(\lambda,\omega) \dd \lambda + \int_{\mathcal{C}_\text{lower}} I_{\mu,\nu}(\lambda,\omega) \dd \lambda \notag\\
  &= \int_{\mathcal{C}_\text{upper}} I_{\mu,\nu}(\lambda,\omega) \dd \lambda - \int_{\mathcal{C}_\text{upper}} I_{\mu,\nu}(\bar{\lambda},\omega) \dd \lambda \notag\\
  &= 2 \,\Re \int_{\mathcal{C}_\text{upper}} I_{\mu,\nu}(\lambda,\omega) \dd \lambda.
\end{align}

\begin{figure}[h]
\centering
\begin{subfigure}{.5\textwidth}
  \centering
  \includegraphics[trim = 0.8in 1.8in 0.8in 2.0in, clip, scale=.25]{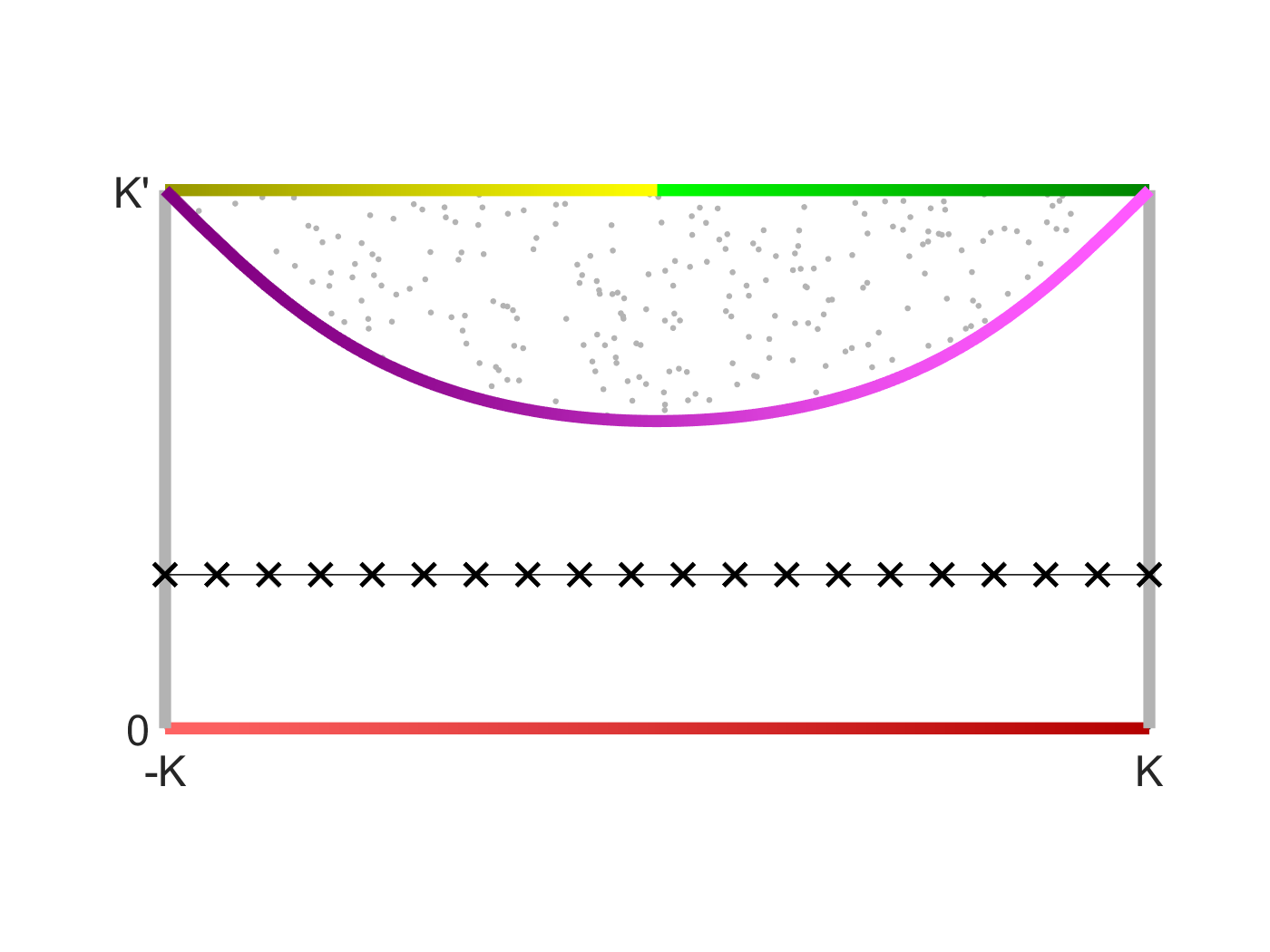}
  \caption{$t$-plane}
  \label{fig:tplane}
\end{subfigure}%
\begin{subfigure}{.5\textwidth}
  \centering
  \includegraphics[trim = 0.8in 1.8in 0.8in 2.0in, clip, scale=.25]{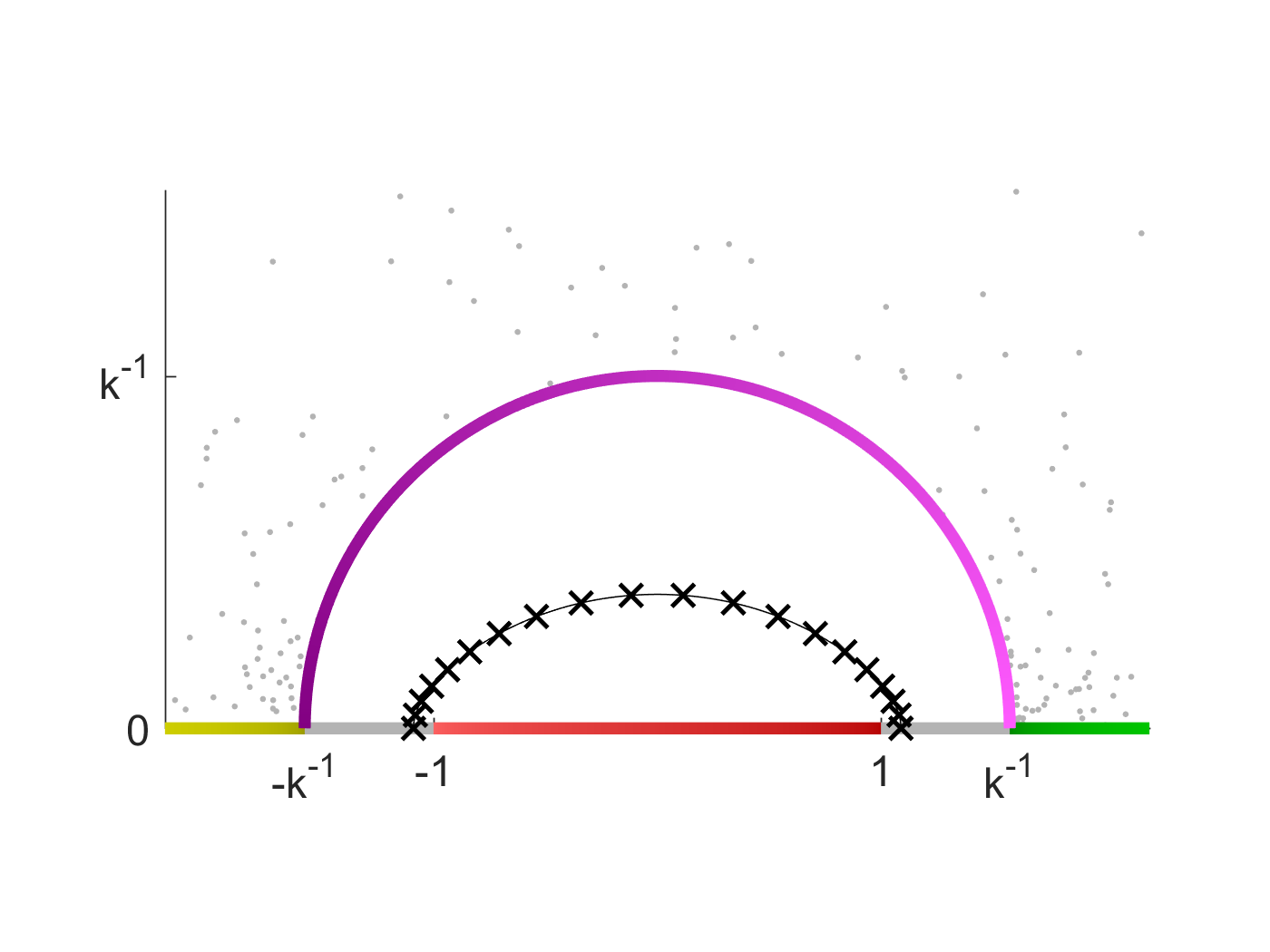}
  \caption{$u$-plane}
  \label{fig:uplane}
\end{subfigure}
\begin{subfigure}{.5\textwidth}
  \centering
  \includegraphics[trim = 0.8in 1.8in 0.8in 2.0in, clip, scale=.25]{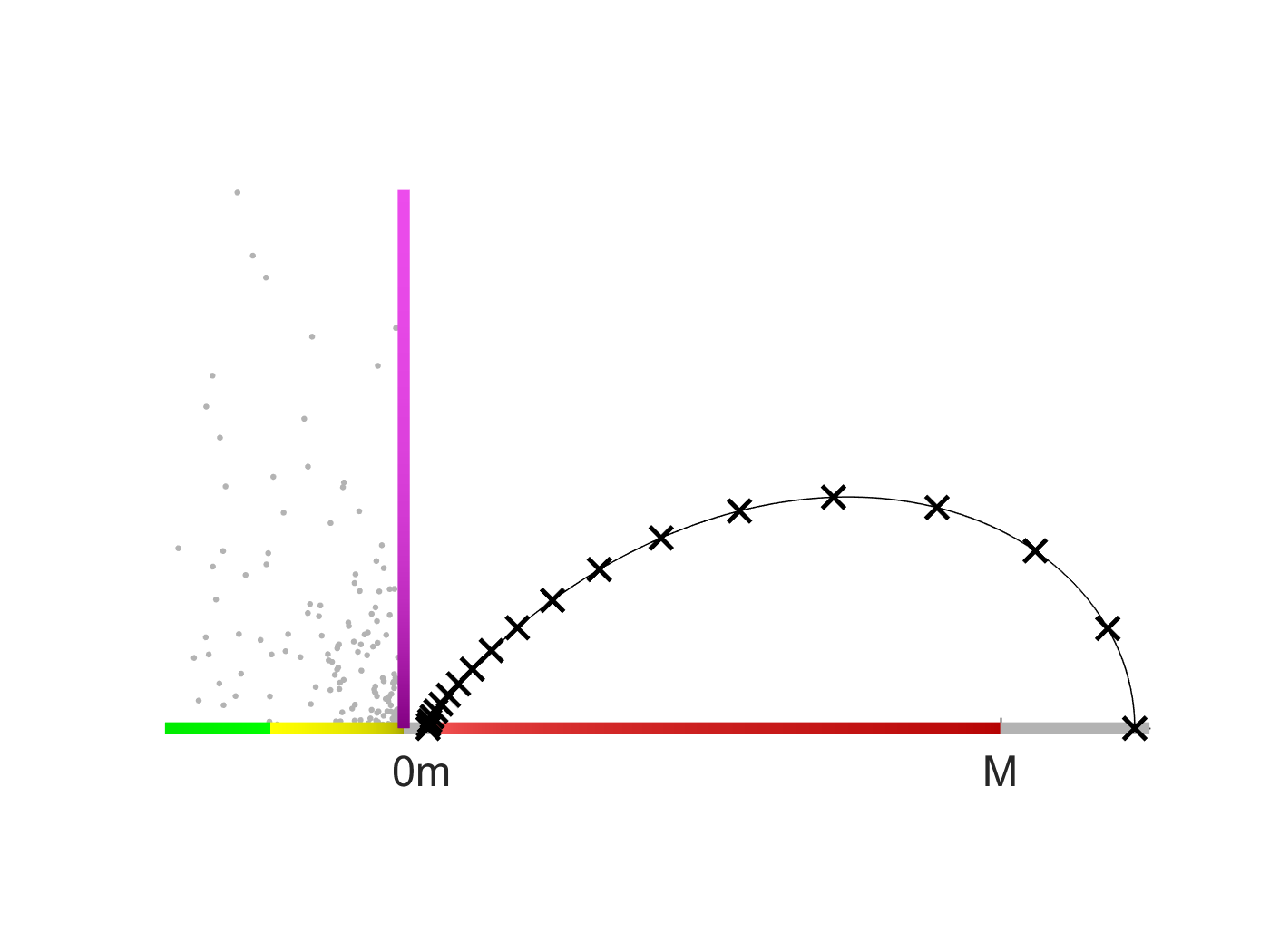}
  \caption{$z$-plane}
  \label{fig:zplane}
\end{subfigure}%
\begin{subfigure}{.5\textwidth}
  \centering
  \includegraphics[trim = 0.8in 1.8in 0.8in 2.0in, clip, scale=.25]{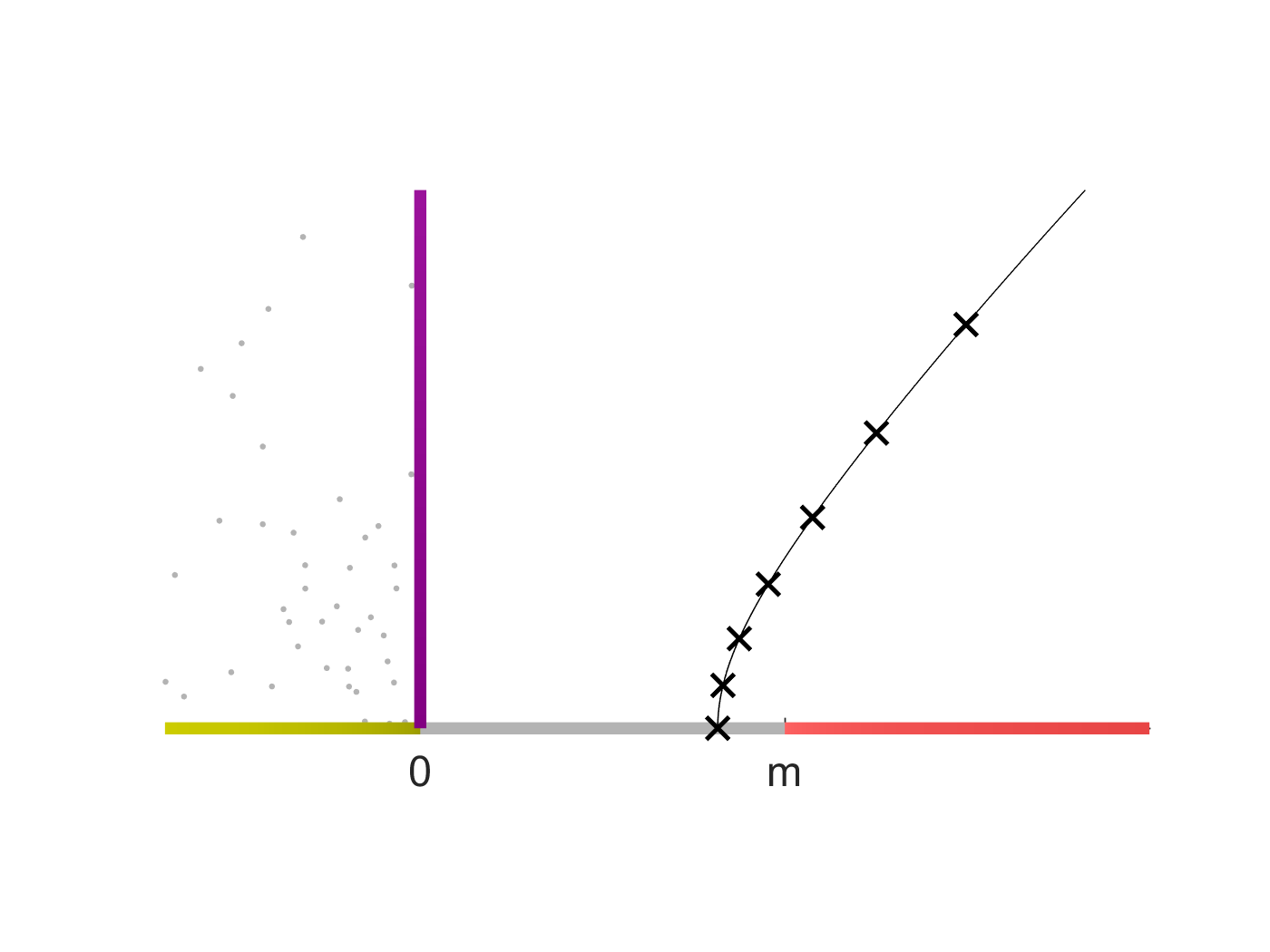}
  \caption{$z$-plane, zoomed in}
  \label{fig:zplaneZoomed}
\end{subfigure}
\caption{These figures show the transformations given in \eqref{tranformttou} and \eqref{transformutoz}.  Coloring is used to help show what is mapped where.  The gray dots were distributed randomly in the region bounded by the purple, yellow, and green curves of the $t$-plane to show that this region maps to the left half $z$-plane.  The red line contains the singularities we wish to encircle.  The purple line and the regions with gray dots contain the singularities we wish to avoid.}
\label{fig:transformations}
\end{figure}

We construct a quadrature rule for Step 3a by using similar ideas as in \cite{hale2008computing, LinLuYingE:09}.  For simplicity, let us assume that $\epsilon_{N_{\occ}} = 0$.  To account for the fact that it is not, we will just have to shift the resulting quadrature points by $\epsilon_{N_{\occ}}$.  For both simplicity of notation and to follow \cite{hale2008computing} more closely, let us define $m = \epsilon_{N_{\occ}+1}- \epsilon_{N_{\occ}}$ to be the energy gap and $M = \epsilon_{N_{\orb}}-\epsilon_{N_{\occ}}$.  Then we map the rectangle with vertices $\pm K$ and $\pm K + iK'$, where $K$ and $K'$ are the complete elliptic integrals~\cite{abramowitz1970handbook}
\begin{align}
K(k) &= \int_0^1 \frac{1}{\sqrt{(1-t^2)(1-k^2t^2)}} \dd t, \\
K'(k) &= K(1-k^2),
\end{align}
to the upper half plane via two consecutive transformations $t \mapsto u \mapsto z$ (see Figure \ref{fig:transformations}).
\begin{align}
u &= \sn(t) = \sn(t|k^2),\qquad\qquad k = \frac{\sqrt{M/m}-1}{\sqrt{M/m}+1}, \label{tranformttou}\\
z &= \sqrt{mM} \left(\frac{k^{-1} + u}{k^{-1} - u}\right), \label{transformutoz}
\end{align}
where $\sn(t|k^2)$ is the Jacobi elliptic function.  The values of $K$ and $K'$ can be found, e.g., via the ellipkkp function in the Schwarz-Christoffel Toolbox for MATLAB \cite{schwartzchristoffeltoolbox}.  The Jacobi elliptic functions $\sn(t)$, $\cn(t)$, and $\dn(t)$ are implemented in the ellipjc function in the same toolbox.  One must be careful when using such functions as there are a few different common conventions for how to parameterize the Jacobi elliptic functions.  We have been using the parameter $k$ while the Schwarz-Christoffel Toolbox uses the parameter $L = -\ln(k)/\pi$.

The idea for the quadrature rule is as follows.  We ultimately wish to construct a quadrature rule in the $z$-plane, but since the function we are integrating is periodic and analytic, one can show via some numerical analysis \cite[Section 4.6.5]{davis1984methods} that we can construct a geometrically convergent quadrature rule by using the trapezoid rule in the $t$-plane.  Essentially, the idea is to map the $z$-plane to a periodic rectangle, apply the trapezoid rule in the periodic rectangle where it is known to converge geometrically, and then map the quadrature points in the $t$-plane back to the $z$-plane to obtain our desire quadrature rule.  The trapezoid rule has the added bonus of having a nice nesting property of the quadrature points, so that we can create an adaptive quadrature rule.  

The numerical analysis tells us that the rate of convergence will be greatest when the quadrature path in the $t$-plane is as far away as possible from all singularities of the function we are integrating.  In order to find the singularities in the $t$-plane, let's first examine them in the $z$-plane.  Due to symmetry, we will only consider the contour and singularities in the upper half plane.  In our case, the function we are integrating is given by \eqref{chitildeJK}.  For a given $\omega \ge 0$, its singularities (in the upper half plane) are $\{\epsilon_j + i\omega\}_{j=1}^{N_{\occ}}$ and $\{\epsilon_k\}_{k=N_{\occ}+1}^{N_{\orb}}$.  We first notice that the singularities we wish to avoid depend on $\omega$.  There is nothing inherently difficult about this, and we could construct a different quadrature rule for each $\omega$.  However, in order to save on computation, we want the quadrature rule to remain the same for each $\omega$.  This way, $\bK(\lambda)$ does not need to be recalculated for each $\omega$.  Therefore, when we construct our quadrature rule, we wish to avoid all such singularities for $\omega \ge 0$.  Since we have assumed that $\epsilon_{N_{\occ}} = 0$, this implies that $\{\epsilon_j + i\omega\}_{\omega\ge 0, j=1,...,N_{\occ}}$ is contained in the left half $z$-plane.  Therefore, it is sufficient for us to say that we wish to avoid the entire left half $z$-plane.

However, when we construct the quadrature rule, we are concerned with the singularities in the $t$-plane.  By inverting the above mappings, we can see in Figure \ref{fig:transformations} that the left half $z$-plane is mapped to the region bounded by the yellow, green, and purple curves in the $t$-plane.  So, it is sufficient for us to avoid this region.  Next, we note that the rest of the singularities in the $z$-plane are contained in the red line.  This line is mapped to the bottom edge of the rectangle in the $t$-plane.  Finally, we wish for our contour in the $z$-plane to encircle the red line in the clockwise direction.  This is achieved by taking a contour in the $t$-plane which goes from the left side of the rectangle to the right side.  It is now clear how to maximize the distance between the singularities and the contour in the $t$-plane.  We must draw our contour in the $t$-plane halfway between the bottom of the rectangle and the lowest point of the purple curve.  This is demonstrated by a black horizontal line with X's in Figure \ref{fig:tplane}.  We apply the trapezoid rule on this line.  The line can be mapped back to the $z$-plane to create a quadrature rule as shown in Figure \ref{fig:zplane}.

Rather than now going into the rigorous details of the above argument, we will simply state the results.  The details and proofs are deferred to the Appendix.  First, the algorithm for Step 3a is summarized in Algorithm \ref{alg:step3a}.
\begin{algorithm}[H]
\caption{Step 3a -- Quadrature rule for contour integral}
\label{alg:step3a}
\begin{algorithmic}[1]
\State Define $m = \epsilon_{N_{\occ}+1}- \epsilon_{N_{\occ}}$ and $M = \epsilon_{N_{\orb}}-\epsilon_{N_{\occ}}$.
\State Compute $k = \frac{\sqrt{M/m}-1}{\sqrt{M/m}+1}$.
\State Compute $I = \frac{1}{2} \int_0^{k^{-1}} \frac{\dd s}{\sqrt{(1+s^2)(1+k^2s^2)}}$ via the midpoint rule with mesh size $h < 1/100$.
\State Define
\begin{equation}
\lambda(t) = \sqrt{mM}\left(\frac{k^{-1} + \sn(t)}{k^{-1} - \sn(t)}\right) + \epsilon_{N_{\occ}},
\end{equation}
as a shift of $z(t)$ to account for the fact that $\epsilon_{N_{\occ}}$ is typically not 0.  The quadrature rule is then found by applying the trapezoid rule (in the variable $t$) to
\begin{equation}
\chi^0_{\mu\nu}(i\omega) = 2 \,\Re\int_{\mathcal{C}_\text{upper}} I_{\mu,\nu}(\lambda,\omega) \dd\lambda = 2\, \Re \int_{-K + iI}^{K+iI} I_{\mu,\nu}(\lambda(t),\omega) \frac{2k^{-1}\sqrt{mM}}{(k^{-1}-\sn(t))^2} \cn(t) \dn(t) \dd t, \label{contourTrap}
\end{equation}
where the contour in the $t$ plane is the horizontal line connecting $-K+iI$ and $K+iI$.
\State Double the number of quadrature points (via the nesting property of the trapezoid rule) until suitable convergence is achieved.
\end{algorithmic}
\end{algorithm}

Next, we state the convergence rate of the proposed quadrature rule, whose proof may be found in the Appendix.

\begin{lemma}
\label{lem:contourConvRate}
Let $N_\lambda$ denote the number of quadrature points used to discretize \eqref{contourTrap} via the trapezoid rule.  Then, for any $M/m > 1$, the error of the quadrature rule is 
\begin{equation}
O\left(\exp\left(\frac{-\pi^2 N_\lambda}{2\log(M/m)+6}\right)\right). \label{contourQuadRuleError}
\end{equation}
\end{lemma}
Therefore, our quadrature rule for the contour integral converges geometrically in the number of quadrature points.  Additionally, the number of points $N_\lambda$ needed for convergence to a specific error tolerance increases only logarithmically as $(\epsilon_{N_\orb}-\epsilon_{N_\occ})/(\epsilon_{N_{\occ+1}}-\epsilon_{N_\occ}) \to \infty$.

Finally, we note that for Step 3b, rather than calculating the trace of the log, it is more efficient to calculate the log of the determinant,
\begin{equation}
\tr[\ln(I-\chi^0(i\omega)v) + \chi^0(i\omega)v] = \ln[\det(I-\chi^0(i\omega)v)] + \tr[\chi^0(i\omega)v].
\end{equation}
This expression is nice for practical computation since it avoids the necessity of calculating the matrix logarithm.  Additionally, the determinant of a matrix may be calculated easily via an LU decomposition, for which there are readily available scalable codes.

\subsection{Step 4 -- Quadrature rule for the frequency integral}

For the $\omega$ integral, we used the following Clenshaw--Curtis quadrature \cite[Eq 3.2]{boyd1987exponentially} on the semi-infinite interval $[0,\infty)$,
\begin{align}
t_m &= \frac{\pi m}{N+1}, \label{CCdefoft} \\
\omega_m &= L \cot^2(t_m/2), \label{omegaIntTransformation} \\
\int_0^\infty f(\omega) \dd \omega &\approx \sum_{m=1}^N W_m f(y_m),
\end{align}
where 
\begin{equation}
W_m = \frac{4L \sin(t_m)}{(N+1)(1-\cos(t_m))^2} \sum_{j=1}^N \frac{1}{j}\sin(j t_m)[1-\cos(j\pi)],
\end{equation}
where $L$ is a parameter that must be chosen (we used $L = 10$).  The value of $L$ can affect the number of grid points needed for convergence, but this dependence is not very sensitive.  A necessary condition for the geometric convergence of this method is that $f(\omega) = O(\omega^{3/2})$ as $\omega \to \infty$ \cite{boyd1987exponentially}.  This is guaranteed by the following lemma.
\begin{lemma}
$\tr\left[\ln(I-\chi^0(i\omega)v) + \chi^0(i\omega)v\right] = O(\omega^{-2})$ as $\omega \to \infty$.
\end{lemma}
\begin{proof}
It is straightfoward to show that $|\chi^0_{\mu\nu}(i\omega)| < c_{\mu\nu}\omega^{-1}$, where $c_{\mu\nu}$ is independent of $\omega$.  This implies $\norm{\chi^0(i\omega)}_F < C \omega^{-1}$.  Then we have
\begin{align}
\norm{\chi^0(i\omega)v}_2 &\le \norm{\chi^0(i\omega)v}_F \notag\\
  &\le \norm{\chi^0(i\omega)}_F \norm{v}_F \notag\\
  &\le C \norm{v}_F \omega^{-1}.
\end{align}
Therefore, for $\omega$ large enough, the eigenvalues of $\chi^0(i\omega)v$ are all less than $1$ in magnitude.  This justifies the Taylor series expansion in the following,
\begin{align}
\left|\tr\left[\ln(I-\chi^0(i\omega)v) + \chi^0(i\omega)v\right]\right| &\le C \sqrt{N_\aux} \norm{\ln(I-\chi^0(i\omega)v) + \chi_0v}_F \notag\\
 &= C \sqrt{N_\aux} \norm{-\frac{1}{2} (\chi^0(i\omega)v)^2 + O(\chi^0(i\omega)v)^3}_F \notag\\
 &\le C \sqrt{N_\aux} \norm{v}_F^2 \omega^{-2} + O(\omega^{-3}), \label{boundOnIntegrand}
\end{align}
where the constant $C$ changes from line to line.
\end{proof}

The use of Clenshaw--Curtis allows a simple and fast converging adaptive quadrature rule since the points of the quadrature rule have a nice nesting property as seen by \eqref{CCdefoft}.

\section{Numerical results}
\label{sec:numerics}

\begin{figure}[h]
\centering
\includegraphics[trim = 0.4in 1.75in 0.1in 1.6in, clip, scale=1.3]{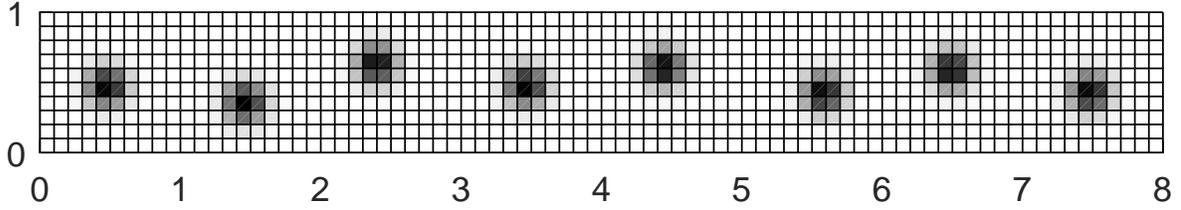}
\caption{Example of our external potential with 8 wells.  White is zero, darker is more negative.}
\end{figure}

Our numerical results use the following as the test problem.  Our two dimensional spatial grid is $10 \times 10N_{\occ}$ equally spaced points.  First, we solve for the KS orbitals of the periodic system with Hamiltonian $H = T + V$, where $T = -\frac{1}{2}\Delta$ is the kinetic energy operator and $V$ is the external potential.  The external potential consists of $N_{\occ}$ Gaussian potential wells, the centers of which are randomly perturbed from the centers of their respective $10\times 10$ box of grid points.  Then the eigenvectors of $H$ are used as the orbitals in the calculation of the RPA correlation energy.

\subsection{Convergence with respect to number of orbitals}

In these tests, we check the convergence of the RPA energy with respect to the number of orbitals used in the calculation.  Figure \ref{fig:convWRTorbitals4} shows the results for a system with 4 electrons (and therefore a maximum of 400 orbitals).  In Figure \ref{fig:convWRTorbitals32}, we have scaled the entire system up by a factor of 8 (maximum of 3200 orbitals) and run the same test.  We can see by comparing the figures that the results are essentially identical.  Both the percentage of orbitals needed in the calculation for a particular error value and the number of auxiliary basis functions (as a percentage of the number of grid points) needed for a particular error level in the ISDF step are nearly the same in the two cases.

In general, one would want to work in a regime where $N_\aux$ is as small as possible while still achieving sufficient accuracy.  Note that if $N_\aux \approx n$, then there is no point in using ISDF and one would be better off using \eqref{chi_0noISDF} instead.  Figure \ref{fig:convWRTorbitals} implies that ISDF is worth doing in the $10^{-2}$ relative error range, where we can see from Figure \ref{fig:convWRTorbitals}, the ISDF yields an $N_\aux$ significantly below $n$.  However, this statement is highly dependent on the number of grid points.  For example, in Figure \ref{fig:convWRTorbitalsCompareGridPts}, we see that the ISDF can be worthwhile all the way down to $10^{-4}$ relative error in the $n=1600$ case.  The reduction of the basis size from $n$ to $N_\aux$ (as a proportion of the number of grid points) is amplified as the number of grid points is increased with all other variables fixed.  This is because, as will be discussed shortly, $N_\aux$ depends on $N_\orb$, not $n$.

In Figure \ref{fig:convWRTorbitalsCompareGridPts}, we fix an external potential and look at the behavior of the algorithm when the number of primal basis functions (grid points) is increased.  Both tests are run with $N_\occ = 4$.  One is run with $n=400$ grid points and the other with $n = 1600$.  Therefore, the maximal number of auxiliary basis functions are 400 and 1600, respectively.  We make two observations about the figure.  First, the number of auxiliary basis functions required depends only on the number of orbitals $N_\orb$ used in the calculation, until saturation occurs.  That is, in both the $n = 400$ and $n = 1600$ tests, $N_\aux$ is essentially the same for $N_\orb \le 200$, at which point the number of auxiliary basis functions starts to max out at 400 in the smaller system.  Second, we note that the error in the two tests is mostly identical for a given number of orbitals $N_\orb$ used in the calculation.  These two observations suggest that the ISDF procedure behaves precisely how one would hope as the number of primal basis functions (grid points) is increased.  That is, (before the auxiliary basis functions max out) the number of auxiliary basis functions and the error in $\ERPAcor$ depends only on $N_\orb$, and not on $n$.  This last observation is, of course, only valid when $n$ is large enough and $\text{tol}$ in the IDSF is small enough that the errors due to the spatial discretization and the ISDF approximation are negligible compared to the error induced by truncating the number of orbitals.

\begin{figure}[h]
\centering
\begin{subfigure}{.5\textwidth}
  \centering
  \includegraphics[scale=.55]{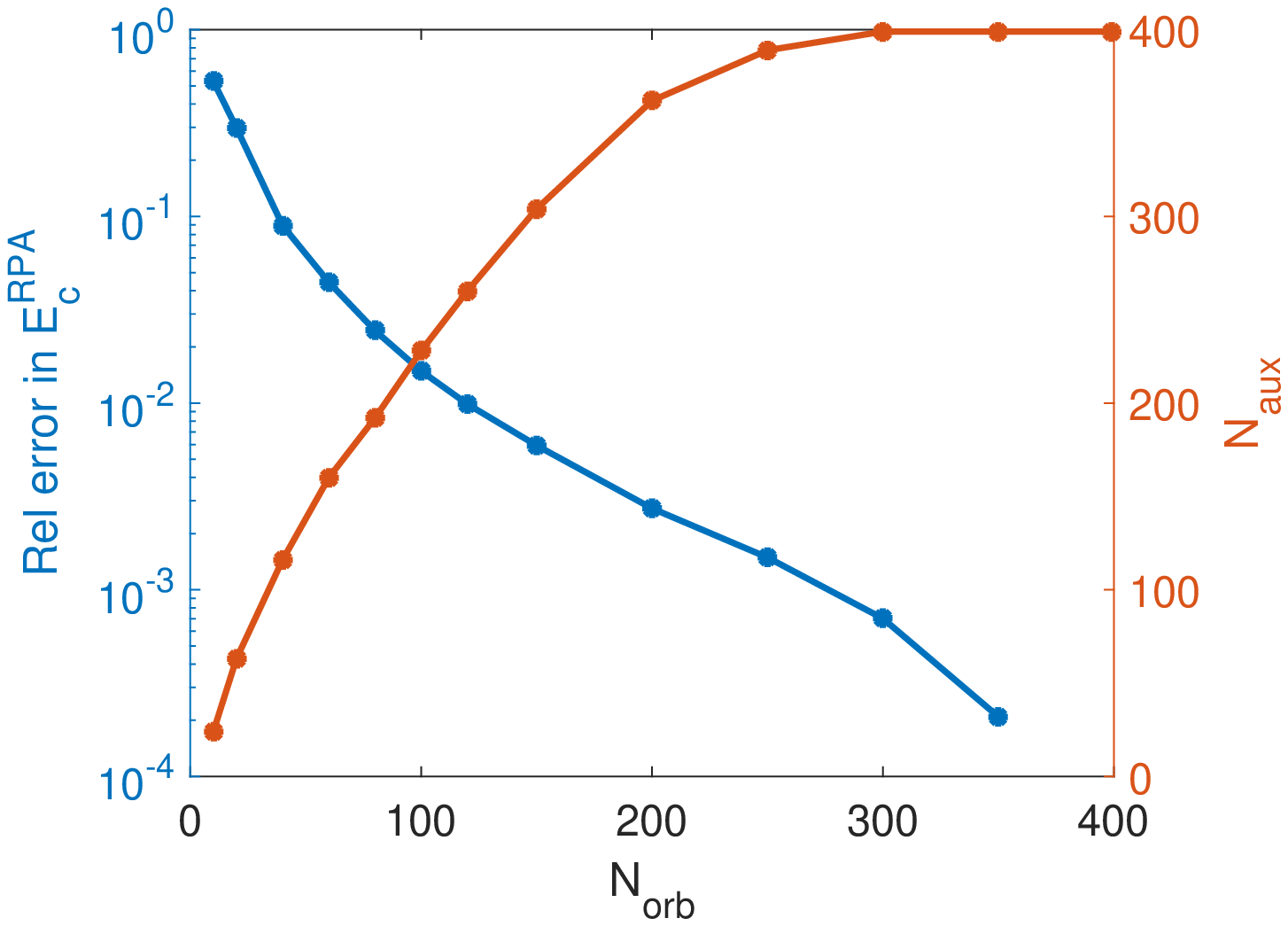}
  \caption{$N_{\occ} = 4$}
  \label{fig:convWRTorbitals4}
\end{subfigure}%
\begin{subfigure}{.5\textwidth}
  \centering
  \includegraphics[scale=.55]{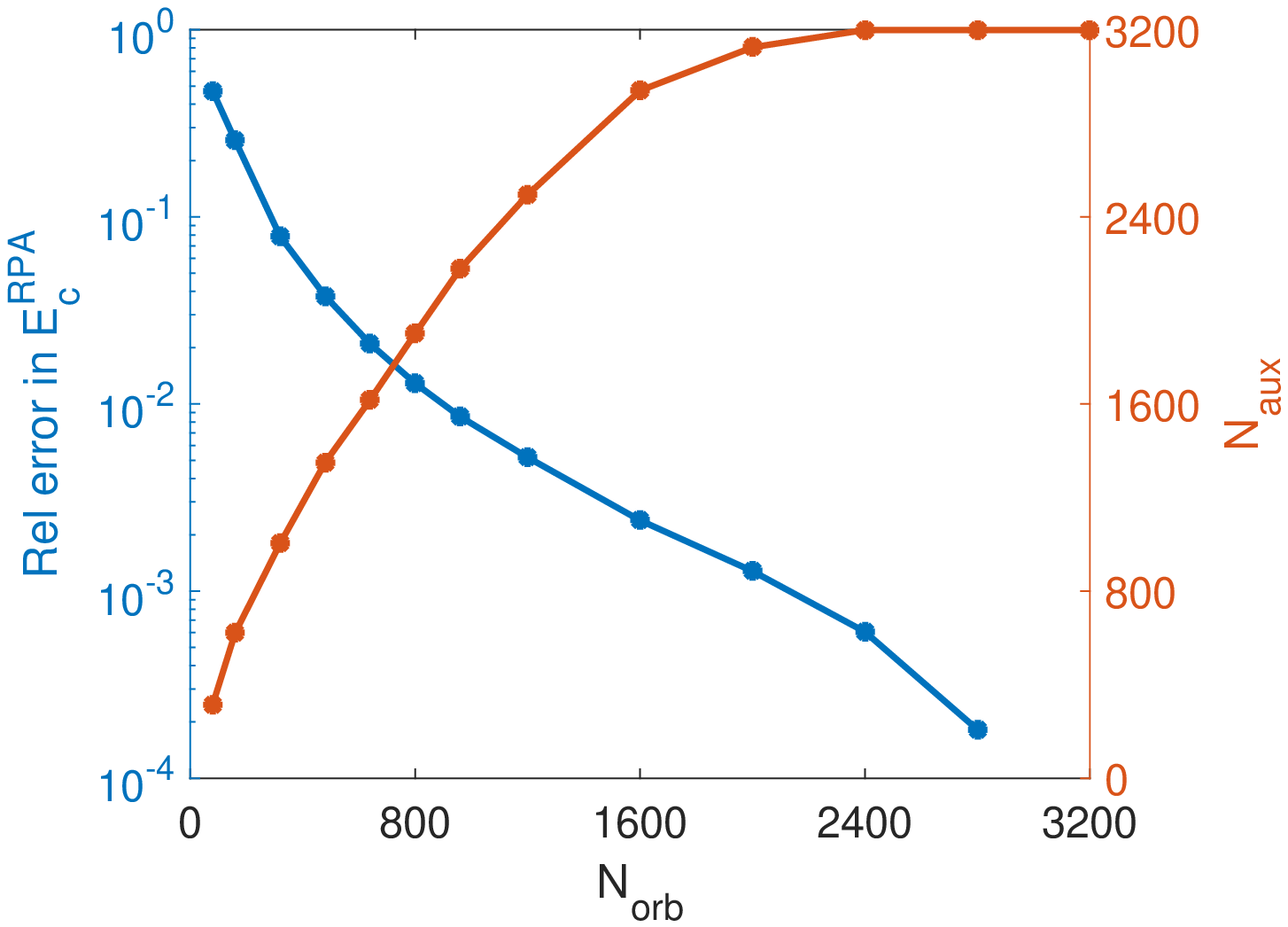}
  \caption{$N_{\occ} = 32$}
  \label{fig:convWRTorbitals32}
\end{subfigure}
\caption{Convergence of the RPA energy with respect to the total number of orbitals used in the calculation.  Both the relative error in $\ERPAcor$ and the number of auxiliary basis functions used in the calculation are plotted.  In the error of each plot, the numerical result with all 400 (3200) orbitals is used as the ``exact'' $\ERPAcor$.  An error tolerance of $\text{tol} = 10^{-4}$ was used in the ISDF.}
\label{fig:convWRTorbitals}
\end{figure}

\begin{figure}[htb]
\centering
\includegraphics[scale=.6]{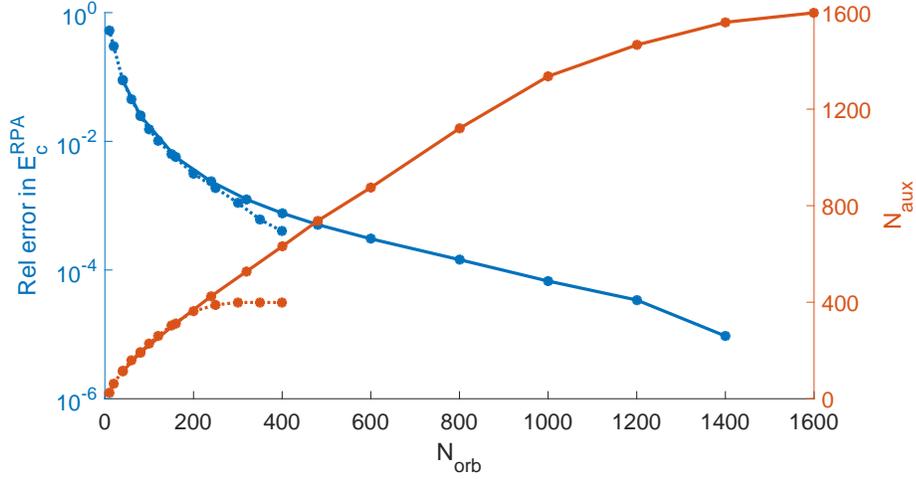}
\caption{Results with $n = 400$ and $n = 1600$ for $N_\occ = 4$ with the same external potential.  Dotted lines are $n = 400$, solid lines are $n = 1600$.  The numerical solution with 1600 orbitals in the $n = 1600$ case is used as the ``exact'' $\ERPAcor$ for purposes of plotting the error.  For the determination of $N_\aux$, we use $\text{tol} = 10^{-4}$ in the ISDF.}
\label{fig:convWRTorbitalsCompareGridPts}
\end{figure}

\subsection{Cubic scaling}

In this test, we show the cubic scaling behavior of the algorithm.  The quartic scaling method using traditional density fitting is also plotted for comparison.  The traditional density fitting requires that we input basis functions, so we use the basis functions $\{P_\mu\}_{\mu = 1}^{N_{\aux}}$ from the ISDF.  The results from Figure \ref{fig:convWRTorbitals} suggest that as we scale up the system size, we can choose the number of orbitals $N_{\orb}$ to use in the calculation as a constant percentage of the number of grid points.  So, we choose $N_{\orb} = 0.2 n$.  We scale the system size up to a maximum of $N_{\occ} = 160$.  We can see in Figure \ref{fig:timing} that the cubic scaling algorithm greatly outperforms the quartic scaling algorithm for large system sizes.

\begin{figure}[h]
\centering
\begin{subfigure}{.5\textwidth}
  \centering
  \includegraphics[scale=.55]{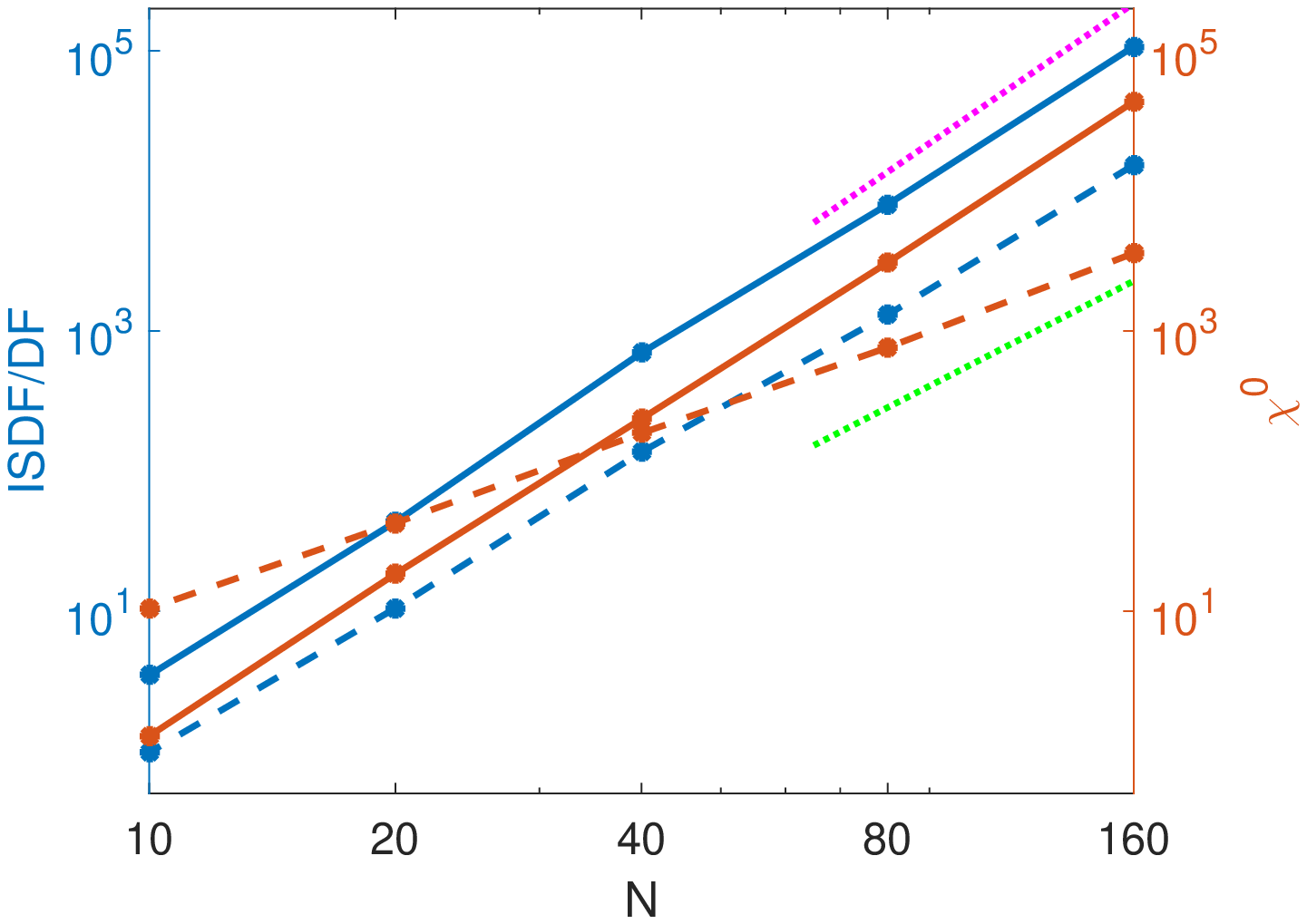}
\end{subfigure}%
\begin{subfigure}{.5\textwidth}
  \centering
  \includegraphics[scale=.55]{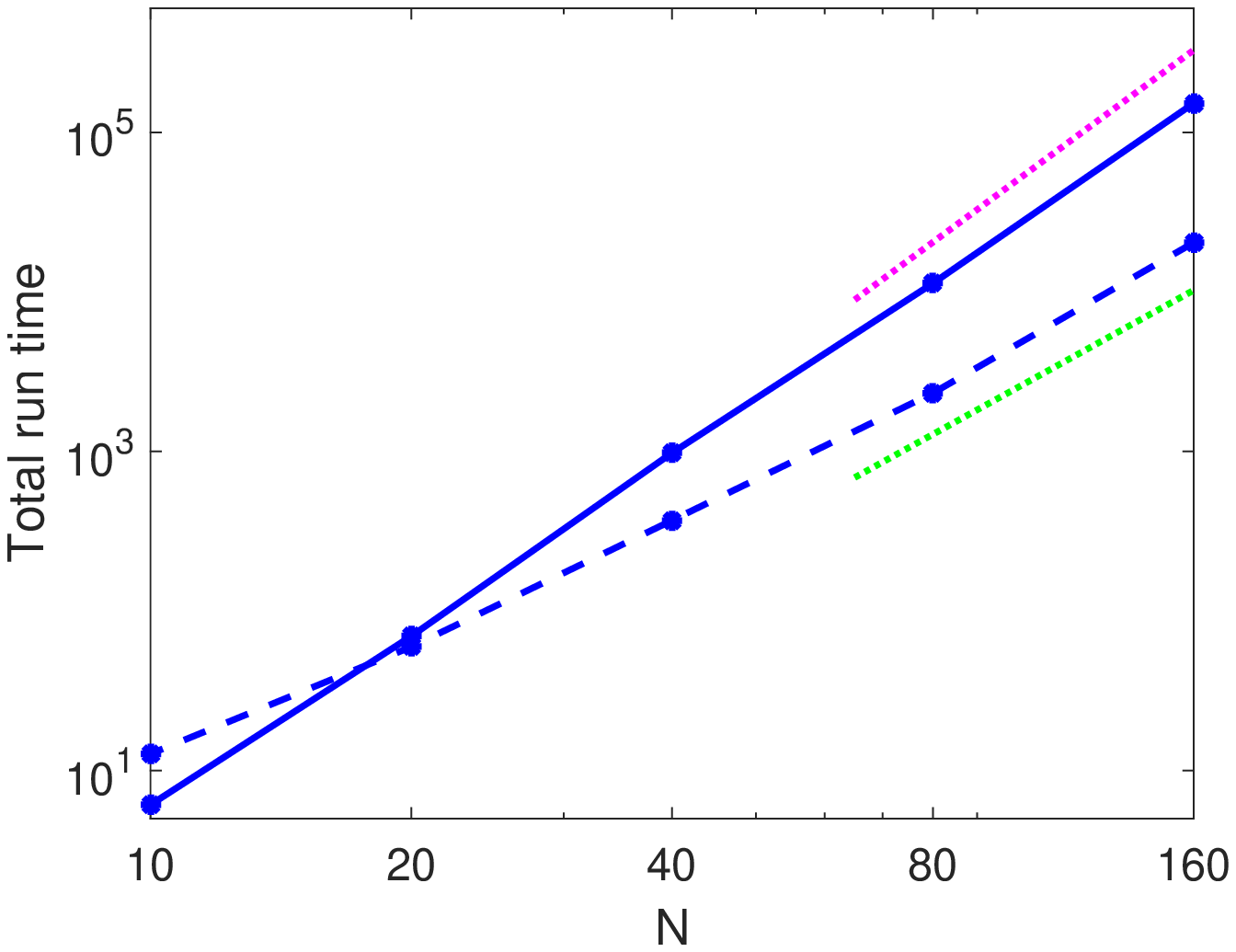}
\end{subfigure}
\caption{The timing results for the quartic scaling method are plotted with solid lines, and the results for the cubic scaling method are plotted with dashed lines.  For reference, the purple and green dotted lines represent the slopes of $N^4$ and $N^3$ respectively.  The left figure compares the time required to calculate $\chi^0$ and the time to perform the respective density fitting schemes for each method.  The right figure compares the total run time to calculate $\ERPAcor$ for each method.}
\label{fig:timing}
\end{figure}

\section{Conclusion}

In this paper, we have presented a new cubic scaling algorithm for the computation of the RPA correlation energy.  The key of the algorithm was to separate the dependence on $j$ and $k$ in the denominator of \eqref{chi_0}.  This allows a natural cubic scaling method.  However, in order to further reduce the computational cost, we employed the ISDF in analogy to how density fitting is used in the quartic algorithm.  Another key idea to keep the computational cost down was to take advantage of the periodic and analytic nature of the function in the contour integral, which resulted in a geometrically convergent nested quadrature rule based on the simple trapezoid rule.

It is worth noting that the algorithm presented is highly parallelizable.  Step 1, the ISDF, is composed of linear algebra routines which can be parallelized.  It is clear that Steps 2 and 3a can be parallelized.  Step 3b is parallelizable using the comment at the end of Section \ref{subsec:contourIntegral}.  There is no need to parallelize Step 4 as it is a simple one dimensional integral, but the Step 3 computations for each quadrature point $\omega_m$ could also be done in parallel.

Future directions include a parallel implementation of the algorithm, as well as implementation into scientific software.  Another direction would be to look into the analytic properties of $\chi^0(i\omega)$.  It would also be interesting to apply the ISDF to the Laplace transform method for cubic scaling RPA algorithms.  We also plan to extend the algorithm presented in this paper to particle-particle RPA (ppRPA) \cite{vanAggelenYangYang:14}.

\appendix 
\section{Derivation of Algorithm~\ref{alg:step3a}}

In this Appendix, we first rigorously derive Algorithm \ref{alg:step3a}. Then we conclude by proving Lemma \ref{lem:contourConvRate}.

Recall from Section \ref{subsec:contourIntegral} that we wish to find the lowest point of the purple curve in the $t$-plane.  We do this now.  First, we note that \eqref{transformutoz} is a M\"{o}bius transformation and therefore its inverse maps the imaginary line to a generalized circle in the $u$-plane.  In particular, it maps the upper half imaginary line in the $z$-plane to the upper semicircle with radius $k^{-1}$ centered at the origin (the purple curve in Figure \ref{fig:uplane}).  To map this semicircle back to the $t$-plane, we note the formula for the inverse of $\sn(t)$, \cite[Chapter 11.3]{beals2010special}
\begin{equation}
\sn^{-1}(u|k^2) = \int_0^u \frac{\dd s}{\sqrt{(1-s^2)(1-k^2s^2)}}. \label{deriveWRTimpart}
\end{equation}
Then we can use basic calculus to minimize $\Im\left[\sn^{-1}(k^{-1}e^{i\theta})\right]$ over $0 \le \theta \le \pi$.
\begin{equation}
\frac{d}{d\theta} \Im\left[\sn^{-1}(k^{-1}e^{i\theta})\right] = \Re\left[ \frac{k^{-1}\sqrt{\cos(2\theta) - 1 - k^{-2} + (k^{-2}-1) e^{-2i\theta}}}{(1+k^{-4} - 2k^{-2} \cos(2\theta))(2-2\cos(2\theta))}\right].
\end{equation}
This expression is 0 if and only if the expression under the radical is nonpositive.  Since the imaginary part of the expression under the radical must be 0, we require $\theta \in \{0, \pi/2, \pi\}$.  $0$ and $\pi$ correspond to the corners of the rectangle, so this means that $\theta = \pi/2$ must give us the minimum imaginary part of points along the purple curve.  In conclusion, we choose our quadrature points in the $t$-rectangle with imaginary part given by
\begin{align}
\frac{1}{2}\Im\left[\sn^{-1}(ik^{-1})\right] &= \frac{1}{2} \Im \int_0^{ik^{-1}} \frac{\dd s}{\sqrt{(1-s^2)(1-k^2s^2)}} \notag\\
  &= \frac{1}{2} \int_0^{k^{-1}} \frac{\dd s}{\sqrt{(1+s^2)(1+k^2s^2)}}.
\end{align}
This integral must be carried out numerically.  However, it is simple and only needs to be done once at the beginning of the algorithm, so we just use the midpoint rule.  We also note that we can easily remove any guess work here by proving a practically useful bound which can be obtained via the standard error analysis for the midpoint rule.  First let
\begin{equation}
g(s) = \frac{1}{(1+s^2)(1+k^2s^2)}.
\end{equation}
Then computation shows
\begin{equation}
g''(s) = \frac{6k^4s^6 + 5k^4s^4 + 2k^4s^2 + 5k^2s^4 - 2k^2s^2 - k^2 + 2s^2 - 1}{[(1+s^2)(1+k^2s^2)]^{5/2}}.
\end{equation}
By noting $0 \le k < 1$ and $0 \le sk \le 1$, we have
\begin{equation}
|g''(s)| \le \frac{13s^2 + 11}{(1+s^2)^{5/2}}. \label{secondDeriveBound}
\end{equation}
Using the fact that the right hand side of \eqref{secondDeriveBound} is decreasing on $[0,\infty)$,
\begin{align}
\left|\int_0^{k^{-1}} g(s) \dd s - h\sum_{j=1}^J f(s_{j+1/2})\right| &\le \frac{1}{24} h^3 \sum_{j=1}^J |f''(\xi_j)| \notag\\
  &\le \frac{1}{24}h^3 \left(13h^2 + 11 + \frac{1}{h} \int_0^{k^{-1}-h} |f''(s)| \dd s \right) \notag\\
  &\le \frac{1}{24}h^3 \left(13h^2 + 11 + \frac{1}{h} \int_0^{k^{-1}} \frac{13s^2 + 11}{(1+s^2)^{5/2}} \dd s \right) \notag\\
  &= \frac{1}{24}h^3 \left( 13h^2 + 11 + \frac{35k^{-3} + 33 k^{-1}}{3h(1+k^{-2})^{3/2}} \right) \notag\\
  &\le \frac{1}{24}h^3 \left( 13h^2 + 11 + \frac{35}{3h} \right),
\end{align}
where the last line uses the fact that the preceding line is a strictly increasing function of $k^{-1}$.  This estimate implies that a mesh size of $1/100$ guarantees an accuracy of $10^{-4}$.  Considering the fact that if the value of this integral is off by a little it will only slightly change the convergence rate, this is sufficiently accurate.

We conclude this discussion of the quadrature rule with some brief analytic results, including the proof of Lemma \ref{lem:contourConvRate}.  First, we note that in a realistic system, $M/m \gg 1$ which implies $k \approx 1$.  This guarantees that using the midpoint rule to calculate $I$ will not require more than about 100 grid points.  Next, we note that $I > K'/4$.  This can be seen by showing that the circle with radius $k^{-1/2}$ centered at the origin in the $u$-plane maps to the horizontal line with imaginary part $K'/2$ in the $t$-plane.  Before proving this statement, let's see why this implies $I > K'/4$.  First, note that $1 < k^{-1/2} < k^{-1}$.  Therefore, the circle with radius $k^{-1/2}$ in the $u$-plane must map between the purple and red curves in the $t$-plane.  This means that the purple curve cannot go any lower than $K'/2$, which implies $I > K'/4$.  To show that the aforementioned circle maps to a line with constant imaginary part, it is enough to show
\begin{equation}
\Im \int_{k^{-1/2}e^{i\theta_1}}^{k^{-1/2}e^{i\theta_2}} \frac{\dd s}{\sqrt{(1-s^2)(1-k^2s^2)}} = \Re \int_{\theta_1}^{\theta_2} \frac{k^{-1/2}\sqrt{2\cos(2\theta) - k^{-1} - k}}{\sqrt{(1-2k^{-1}\cos(2\theta) + k^{-2})(1-2k\cos(2\theta) + k^2)}} \dd\theta,
\end{equation}
is equal to 0 for all $0 \le \theta_1 \le \theta_2 \le \pi$.  It is easily verified that (for $0 < k < 1$) the denominator on the right is always positive and the numerator is always a pure imaginary number.  Therefore, the integrand is always purely imaginary, which proves the claim.  Finally, the imaginary part of the line is $K'/2$ since $\sn^{-1}(ik^{-1/2}) = iK'/2$ \cite[Table 22.5.2]{NIST:DLMF}.

Now following \cite{hale2008computing} and using the fact that $I > K'/4$, we have that for any $M/m > 1$, the error of the quadrature rule is 
\begin{equation}
O\left(\exp\left(\frac{-\pi^2 N_\lambda}{2\log(M/m)+6}\right)\right).
\end{equation}

\bibliography{references}
\bibliographystyle{abbrv}

\end{document}